\documentclass[12pt]{amsart}
\usepackage{amssymb}
\usepackage{color}
\usepackage{amsmath,epic,curves,amscd}
\usepackage[english]{babel}
\usepackage{graphicx}
\usepackage{comment}
\usepackage{appendix}
\usepackage{mathdots}
\usepackage[all]{xy}
\usepackage{mathtools}
\usepackage{lscape}
\usepackage{stmaryrd}
\usepackage{mathabx}
\newtheorem{claim}{}[section]
\newtheorem{theorem}[claim]{Theorem}

\newtheorem{lemma}[claim]{Lemma}
\newtheorem{proposition}[claim]{Proposition}
\newtheorem{corollary}[claim]{Corollary}

\theoremstyle{remark}

\renewenvironment{proof}{\noindent{\it Proof. \hskip0pt}}
                      {$\square$\par\medskip}
\allowdisplaybreaks \textwidth 14.9 true cm \textheight 22.9 true cm
\begin{document}
\baselineskip 6.0 truemm
\parindent 1.5 true pc

\newcommand\lan{\langle}
\newcommand\ran{\rangle}
\newcommand\tr{{\text{\rm Tr}}\,}
\newcommand\ot{\otimes}
\newcommand\ttt{{\text{\sf t}}}
\newcommand\rank{\ {\text{\rm rank of}}\ }
\newcommand\choi{{\rm C}}
\newcommand\dual{\star}
\newcommand\flip{\star}
\newcommand\cp{{\mathcal C}{\mathcal P}}
\newcommand\ccp{{\mathcal C}{\mathcal C}{\mathcal P}}
\newcommand\pos{{\mathcal P}}
\newcommand\tcone{T}
\newcommand\mcone{K}
\newcommand\superpos{{\mathbb S\mathbb P}}
\newcommand\blockpos{{\mathcal B\mathcal P}}
\newcommand\jc{{\text{\rm JC}}}
\newcommand\dec{{\mathcal D}{\mathcal E}{\mathcal C}}
\newcommand\ppt{{\mathcal P}{\mathcal P}{\mathcal T}}
\newcommand\xxxx{\bigskip\par ================================}
\newcommand\join{\vee}
\newcommand\meet{\wedge}
\newcommand\ad{{\text{\rm Ad}}\,}
\newcommand\HS{{\text{\rm HS}}}
\newcommand\sr{{\text{\rm SR}}\,}
\newcommand\e{\varepsilon}
\newcommand\re{{\text{\rm Re}}\,}
\newcommand\diag{{\text{\rm Diag}\,}}
\newcommand\id{{\text{\rm id}}}
\newcommand\SN{{\text{\rm SN}}\,}
\newcommand\inte{{\text{\rm int}}\,}
\newcommand\im{{\text{\rm Im}}\,}
\newcommand\SR{{\text{\rm SR}}\,}
\newcommand\Aff{{\text{\rm Aff}}\,}
\newcommand\Gr{{\text{\rm Gr}}}
\newcommand\dist{{\text{\rm dist}}}

\title{Local geometry for Schmidt number witnesses}

\author{Young-Hoon Kiem and Seung-Hyeok Kye}
\address{School of Mathematics, Korea Institute for Advanced Study, 85 Hoegiro, Dongdaemun-gu, Seoul 02455, Korea}
\email{kiem@kias.re.kr}
\address{Department of Mathematics and Institute of Mathematics, Seoul National University, Seoul 08826, Korea}
\email{kye@snu.ac.kr}

\thanks{Partially supported by NRF-RS-2025-23323242, Korea.}

\begin{abstract}
Suppose that $F_E$ is the face of the convex set of all $m\ot n$ bi-partite
states which consists of states with ranges contained in a subspace
$E$. By the result in [Phys. Rev. A {\bf 112} (2025), 032426], there exist
Schmidt number $\ell$ witnesses outside of the face $F_{E^\perp}$ if and only if $\ell\le\mu$,
where $\mu$ is the minimum of Schmidt ranks of nonzero vectors in $E$.
In this paper, we show that every neighborhood of the projection state
located at the center of $F_{E^\perp}$ contains Schmidt number $\ell$ witnesses
if and only if $\ell\ge\mu$. In order to understand the local shape of the convex sets of
$k$-blockpositive matrices around this projection state, we introduce principal supporting
hyperplanes, and show that they correspond to vectors in $E$ with Schmidt rank $\le k$.
Using techniques from algebraic geometry, we show that the inclusions between these convex sets
are detected by principal supporting hyperplanes through the projection state for generic subspaces $E$,
while for non-generic subspaces they may not be detected by any supporting hyperplane.
\end{abstract}
\maketitle

%%%%%%%%%%%%%%%%%%%%%%%%%%%%%%%%%%%%%%%%%%%%%%%%%%%%%%%%%%%%%%%%%%%%%%%%%%%%%%%%%%%%%%%%%%%%%%%%%%%%%
\section{Introduction}

The Schmidt number \cite{terhal-schmidt} of a bi-partite quantum state is considered as an
important tool to measure the degree of entanglement; the more
entangled a state is, the higher Schmidt number it has, and it is
separable if and only if its Schmidt number is just one. It is a
very difficult problem to determine the Schmidt number of a state,
and we need  Schmidt number $k$ witnesses \cite{SBL_2001} for this purpose,
which lie on the set difference $\blockpos_{k-1}\setminus\blockpos_k$,
where $\blockpos_k$ denotes the convex set of all $k$-blockpositive matrices
\cite{{jam_72},{sko-thesis}} with trace one.
Recall that a bi-partite Hermitian matrix
is $k$-blockpositive if and only if it is the Choi matrix \cite{choi75-10} of a
$k$-positive map \cite{stine} on matrix algebras.

We consider the convex set $\mathcal D$ of all bi-partite states, or
equivalently density matrices in $M_m\ot M_n$, on the affine space
$\mathcal H$ of matrices with trace one in the real vector space $(M_m\ot M_n)^{\rm h}$ of Hermitian matrices in $M_m\ot M_n$.
Then it is clear that Schmidt number witnesses are located
outside of $\mathcal D$ on the affine space $\mathcal H$.
We recall that the boundary of a convex set is completely partitioned
by the relative interiors of nontrivial faces, and so outside of the convex set is also partitioned
by nontrivial faces \cite{han_kye_2025_a}. For a given subspace $E$
in $\mathbb C^m\ot\mathbb C^n$, it is well known that the set
$F_E$ of all states whose ranges are contained in $E$ is a face of the convex set $\mathcal D$,
and every face of $\mathcal D$ arises in this way.

For a nonzero subspace $E$ of $\mathbb C^m\ot\mathbb C^n$, we denote by
$$
\mu(E)=\min\{\SR|\xi\ran: 0\neq|\xi\ran\in E\}
$$
the minimum of Schmidt ranks of nonzero vectors in $E$, where $\SR|\xi\ran$ denotes the
Schmidt rank of $|\xi\ran\in\mathbb C^m\ot\mathbb C^n$.
It was shown in \cite{han_kye_2025_a} that there exist Schmidt
number $\ell$ witnesses outside of the face $F_{E^\perp}$ if and only if
$2\le\ell\le\mu(E)$. A subspace $E$ with $\mu(E)=k$ is called {\sl exactly $(k-1)$-entangled} in
\cite{han_kye_2025_a}, which classifies completely entangled spaces
\cite{Parthasarathy_2004}. See \cite{lovitz_jognson} for the notion of $k$-entangled subspaces.
Especially, there exists no entanglement witness outside of
$F_{E^\perp}$ at all if and only if $E$ contains a product vector.
See \cite{kye-cambridge} for the related result for positive maps,
when $E$ is one dimensional. On the other hand, there exist
Schmidt number $\ell$ witnesses outside of $F_{E^\perp}$ for every
$\ell=2,3,\dots, m\meet n$ if and only if every nonzero vector in $E$
has the maximal Schmidt rank $m\meet n$, where $m\meet n=\min\{m,n\}$. If $m=n$, this happens only when $E$
is a one dimensional space spanned by a vector with the maximal
Schmidt rank. Choi matrices of Tomiyama maps
\cite{{tom_85}} give rise to typical examples for such a case, which
distinguish Schmidt numbers of isotropic states
\cite{terhal-schmidt}. See {\sc Figure 1}.

\begin{figure}
\setlength{\unitlength}{0.6 truecm}
\begin{picture}(17,4.9)
\thinlines

\dottedline{0.05}(0,2)(17,2)

\put(1,2){\circle*{0.2}}%A
\qbezier(1.25,1.0)(0.75,2)(1.25,3)
\put(1.3,0.4){$\blockpos_1$}

\put(3,2){\circle*{0.2}}%A
\qbezier(3.25,1.0)(2.75,2)(3.25,3)
\put(3.3,0.4){$\blockpos_2$}

\put(1.8,2.8){$\downarrow$}
\put(1,3.5){SN 2 W}
\qbezier(1,2)(2,3)(3,2)

\put(6.6,2){\circle*{0.2}}%A
\qbezier(6.85,1.0)(6.35,2)(6.85,3)
\put(6.8,0.4){$\blockpos_{\mu-1}$}

\put(9,2){\circle*{0.2}}%A
\qbezier(9.25,1.0)(8.75,2)(9.25,3)
\put(9.2,0.4){$\blockpos_\mu$}

\put(7.6,2.8){$\downarrow$}
\put(7,3.5){SN $\mu$ W}
\qbezier(6.6,2)(7.8,3)(9,2)

\qbezier(9,2)(10.5,1)(12,1)
\put(12.2,0.8){$\mathcal D$}
\qbezier(13,1)(14.5,1)(16,2)
\put(16,2){\circle*{0.2}}%A
\put(9,2.2){$F_{E^\perp}$}
\qbezier(15.75,1)(16.25,2)(15.75,3)
\put(16,2.2){$F_E$}

\put(12,2){\circle*{0.2}}%A
\put(12,2.3){$\varrho_*$}

\put(4.5,2){\circle*{0.2}}%A
\put(5,2){\circle*{0.2}}%A
\put(5.5,2){\circle*{0.2}}%A

\put(4.7,3.8){\circle*{0.2}}%A
\put(5.2,3.8){\circle*{0.2}}%A
\put(5.7,3.8){\circle*{0.2}}%A

\end{picture}
\caption{This picture shows how Schmidt number $\ell$ witnesses for $\ell=2,\dots,\mu$
are located on the line through the maximally mixed state $\varrho_*$
and projection states on the faces $F_E$ and $F_{E^\perp}$, where $\mu=\mu(E)$
is the minimum of Schmidt ranks of nonzero vectors in $E$.}
\end{figure}
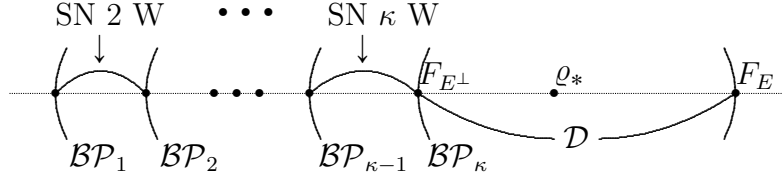

In this paper, we pay attention to what happens around the projection state
$\varrho_{E^\perp}=\frac 1{\dim {E^\perp}}P_{E^\perp}$ with the projection $P_{E^\perp}$ onto ${E^\perp}$, which is
located at the  center of the face $F_{E^\perp}$. If $\mu(E)\le k$, then
the face $F_{E^\perp}$ of $\mathcal D$ is contained in the
boundary of the convex set $\blockpos_{k}$, and there is no Schmidt number $k+1$ witness
outside of $F_{E^\perp}$. We ask whether there exist Schmidt number $\ell$ witnesses with $\ell>k$ arbitrarily close to $\varrho_{E^\perp}$.
Corollary \ref{cor_iff} answers this question completely; for every nonzero proper subspace $E$ and $\ell=2,\dots,m\meet n$,
every neighborhood of $\varrho_{E^\perp}$ contains Schmidt number $\ell$ witnesses if and only if $\ell\ge\mu(E)$.
Together with the result in \cite{han_kye_2025_a} mentioned above, Schmidt number $\ell$ witnesses
appear outside of $F_{E^\perp}$ exactly when $\ell\le\mu(E)$, and arbitrarily close to its center
$\varrho_{E^\perp}$ exactly when $\ell\ge\mu(E)$.

In order to understand the local shape of $\blockpos_k$ around $\varrho_{E^\perp}$,
we consider the supporting hyperplanes of
$\blockpos_k$ through $\varrho_{E^\perp}$ as in
\cite{han_kye_hyper}, and introduce the notion of {\sl principal}
supporting hyperplanes among them which determine the local shape of
$\blockpos_k$ around $\varrho_{E^\perp}$ to the first order. It turns out that they
correspond to the extreme rays of the dual face of $F_{E^\perp}$
which consist of all (unnormalized) states of Schmidt number $\le k$
with ranges in $E$. Therefore, the
local shape of $\blockpos_k$ around $\varrho_{E^\perp}$ is determined by
vectors in $E$ with Schmidt rank $\le k$. In this context, we
consider
$$
\begin{aligned}
S(E,k)&:=\{|\xi\ran\in E: \SR |\xi\ran \le k\},\\
S^\prime(E,k)&:=\{|\xi\ran\in E: \SR |\xi\ran = k\},
\end{aligned}
$$
for a subspace $E$ of $\mathbb C^m\ot\mathbb C^n$ and $k=1,2,\dots,m\meet n$.
We show in Proposition \ref{pro_detect} that the inclusion $\blockpos_k\supsetneqq\blockpos_{k+1}$
is detected by a supporting hyperplane through $\varrho_{E^\perp}$ if and only if $S^\prime(E,k+1)$ is nonempty,
in which case it is already detected by a principal one.

For a given natural number $d=1,2,\dots,mn$, there exists a unique natural number
$\kappa_d\in\{1,2,\dots,m\meet n\}$ satisfying
\begin{equation}\label{def_kappa}
(m-\kappa_d)(n-\kappa_d)<d\le (m-(\kappa_d-1))(n-(\kappa_d-1)).
\end{equation}
For a generic subspace $E$ of dimension $d$, it is well known
 that $\mu(E)=\kappa_d$. See \cite{CMW08}.
In this paper, for generic $d$ dimensional subspaces $E$ of $\mathbb C^m\ot\mathbb C^n$
we show that $S^\prime(E,k)$ is nonempty for every $k\ge \kappa_d$.
In fact, we find the cardinality and dimension of $S(E,k)$ for $k\ge \kappa_d$ for generic cases,
by standard techniques from algebraic geometry.
Consequently, for generic subspaces $E$ with $\dim E=d$, the inclusions
$\blockpos_\ell\supsetneqq\blockpos_{\ell+1}$ for $\ell\ge\kappa_d$ are detected by principal supporting hyperplanes
through $\varrho_{E^\perp}$. For non-generic subspaces this may fail, as we see in examples.
In such cases, the Schmidt number witnesses around $\varrho_{E^\perp}$ come from the non-linear geometry of the
boundary, and we construct them by a direct perturbation argument.

In the next section, we define the notion of {\sl principality} of supporting hyperplanes of a convex set
which is essential to determine the local shape of the convex set, and show that principal supporting hyperplanes
correspond to extreme rays of the dual convex cones. In Section 3, we apply the result to the convex set $\blockpos_k$ of
$k$-blockpositive matrices and analyze $S(E,k)$ for generic subspaces $E$.
After we examine a couple of examples in Section 4, we prove the main result for arbitrary subspaces in Section 5,
and conclude the paper with a summary and a question in the last section.

The authors are grateful to Kyung Hoon Han for fruitful discussions on the topics.

%%%%%%%%%%%%%%%%%%%%%%%%%%%%%%%%%%%%%%%%%%%%%%%%%%
\section {Principal supporting hyperplanes}
%%%%%%%%%%%%%%%%%%%%%%%%%%%%%%%%%%%%%%%%%%%%%%%%%%

Suppose that $C$ is a closed convex set in a finite dimensional affine space with
nonempty interior. We recall that a hyperplane $H$ in the affine
space is a {\sl supporting hyperplane} of $C$ when $C$ is located in
one of the closed half-spaces, say $H^+$, determined by $H$, and has
nonempty intersection with $C$.
It is well known that for every boundary point $x$ of $C$ there
exists a supporting hyperplane of $C$ through the point $x$.
When $H$ is a supporting hyperplane of a closed convex set $C$, it is clear that $H\cap C$ is a face of $C$.
If three supporting hyperplanes $H_0,H_1$ and $H_2$ of $C$ satisfy
$C\subset H_1^+\cap H_2^+\subset H_0^+$ then  $H_0$ does not
contribute to determine the local shape of $C$ around $C\cap H_0$.
We call a supporting hyperplane $H$ of $C$ {\sl principal} when $H$
satisfies the following condition; if two supporting hyperplanes
$H_1$ and $H_2$ satisfy $C\subset H_1^+\cap H_2^+\subset H^+$ then
$H=H_1$ or $H=H_2$.
See {\sc Figure 2}.

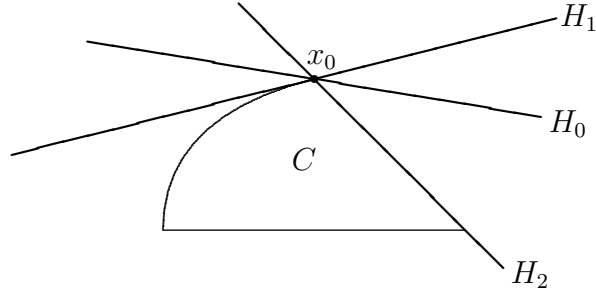
\begin{figure}
\begin{center}
\setlength{\unitlength}{1 truecm}
\begin{picture}(5,4)

\put(3,3){\circle*{0.1}}
\put(2.9,3.2){$x_0$}
\put(2.7,1.8){$C$}
\thinlines
\qbezier(1,1)(1,2.5)(3,3)
\drawline(1,1)(5,1)
\thicklines
\drawline(-1,2)(3,3)
\drawline(3,3)(6.2,3.8)
\drawline(2,4)(3,3)
\drawline(3,3)(5.5,0.5)
\drawline(3,3)(0,3.5)
\drawline(3,3)(6,2.5)

\put(6.25,3.7){$H_1$}
\put(6.1,2.3){$H_0$}
\put(5.6,0.3){$H_2$}

\end{picture}
\end{center}
\caption{
This picture shows three supporting hyperplanes $H_0$, $H_1$ and $H_2$ of the convex set $C$
through the boundary point $x_0$. Two supporting hyperplanes
$H_1$ and $H_2$ are principal, and determine the local shape of $C$ approximately around $x_0$.
The supporting hyperplane $H_0$, which is not principal,
has no contribution to the local shape of $C$.
}
\end{figure}

Suppose that $V$ and $W$ are finite dimensional real vector spaces with a
non-degenerate bilinear pairing $\lan\ , \ \ran$ on $V\times W$. For an
arbitrary nonempty subset $C$ of $V$, we define
$$
C^\circ=\{y\in W: \lan x,y\ran\ge 0\ {\text{\rm for every}}\ x\in
C\}.
$$
Then $C^\circ$ is a closed convex cone in $W$, that is, $C^\circ$ is
a closed set in $W$ which is closed under addition and nonnegative scalar multiplications. For a
subset $D$ of $W$, the convex cone $D^\circ$ in $V$ is also defined
in the same way. It is well known that $C^{\circ\circ}$ is the
smallest closed convex cone containing $C$, and so a convex cone $C$
is closed if and only if $C=C^{\circ\circ}$.

From now on, we suppose that $C$ is a closed convex cone in $V$ with nonempty interior. Noting that every
supporting hyperplane of a closed convex cone should be through
zero, we define
$$
H_x=\{y\in W: \lan x,y\ran=0\},\qquad
H_x^+=\{y\in W:\lan x,y\ran\ge 0\},
$$
for $x\in V$.
%, and $$H_z=\{x\in V: \lan x,z\ran=0\},\qquad H_z^+=\{x\in V:\lan x,z\ran\ge 0\},$$ for $z\in W$.
We note that $C^\circ\subset H_x^+$ if and only if $x\in C^{\circ\circ}$ if and only if $x\in C$,
and so every supporting hyperplane of $C^\circ$ is of the form $H_x$ for $x\in C$.
We also have $C^\circ=\bigcap_{x\in C}H_x^+$ by definition.
We recall \cite{eom-kye} that $x\in C$ is an interior point of $C$ if and only if
$\lan x,y\ran>0$ for every nonzero $y\in C^\circ$ if and only if $H_x\cap C^\circ=\{0\}$.
Therefore, $x\in C$ gives rise to a {\sl nontrivial} supporting hyperplane of $C^\circ$
satisfying the condition  $H_x\cap C^\circ\supsetneqq \{0\}$
if and only if $x$ is a nonzero boundary point of $C$.

Recall that nonzero $x\in C$ generates an extreme ray of $C$ if and only if
$x=x_1+x_2$ with $x_1,x_2\in C$ implies that both $x_1$ and $x_2$ are nonnegative scalar multiples of $x$.
If the convex cone $C$ satisfies $C\cap (-C)=\{0\}$, then it is also equivalent to the condition
that $x=x_1+x_2$ with $x_1,x_2\in C$ implies that either $x_1$ or $x_2$ is a positive scalar multiple of $x$.

\begin{proposition}\label{pro_prinSH}
Suppose that $C$ is a closed convex cone in a finite dimensional real vector space $V$ satisfying $C\cap(-C)=\{0\}$.
For a given nonzero $x\in C$, the supporting hyperplane $H_x$ of
$C^\circ$ is principal if and only if $x$ generates an extreme ray
of $C$.
\end{proposition}

\begin{proof}
Suppose that $H_x$ is a principal supporting hyperplane.
In order to show that $x$ generates an extreme ray, we suppose that
$x=x_1+x_2$ with $x_1,x_2\in C$. If $x_1=0$ or $x_2=0$, then there is nothing to prove.
Otherwise, $H_{x_1}$ and $H_{x_2}$ are supporting
hyperplanes of $C^\circ$, and we have $C^\circ\subset H_{x_1}^+\cap H_{x_2}^+\subset H_x^+$, which
implies $H_x=H_{x_1}$ or $H_x=H_{x_2}$. If $H_x=H_{x_1}$, then $x_1=\lambda x$ for a nonzero real number $\lambda$,
and $\lambda>0$ because $C\cap(-C)=\{0\}$. The case $H_x=H_{x_2}$ is similar.
Therefore, $x_1$ or $x_2$ is a positive
scalar multiple of $x$.

For the converse, suppose that $x$ generates an extreme ray of $C$.
We also suppose that $H_{x_1}$ and $H_{x_2}$ are supporting
hyperplanes of $C^\circ$ satisfying $C^\circ\subset
H_{x_1}^+\cap H_{x_2}^+\subset H_x^+$, where $x_1,x_2\in C$.
We note that $(H_x^+)^\circ$ is the
ray $\{\lambda x: \lambda\ge 0\}$ generated by $x$. From this, it is
easily seen that $(H_{x_1}^+\cap H_{x_2}^+)^\circ=(H_{x_1}^+)^\circ
+  (H_{x_2}^+)^\circ$, and so we have $(H_x^+)^\circ\subset
(H_{x_1}^+)^\circ + (H_{x_2}^+)^\circ$. Therefore, we see that
$x=\lambda_1x_1+\lambda_2x_2$ with $\lambda_i\ge 0$ for $i=1,2$.
Extremality of $x$ implies that
%both $\lambda_1x_1$ and $\lambda_2x_2$ are nonnegative scalar multiples of $x$, and at least one of them is nonzero. Hence
$x_1$ or $x_2$ is a positive scalar multiple of $x$,
from which we have $H_x=H_{x_1}$ or $H_x=H_{x_2}$. Therefore, we conclude that
$H_x$ is a principal supporting hyperplane.
\end{proof}

By the duality $C=C^{\circ\circ}$, we may exchange the roles of $C$ and $C^\circ$
in Proposition \ref{pro_prinSH}. Therefore, for a nonzero $z\in C^\circ$, the supporting hyperplane $H_z$ of $C$
is principal if and only if $z$ generates an extreme ray of $C^\circ$.

Keeping the hyperplane $\mathcal H$ in the space $(M_m\ot M_n)^{\rm h}$ given by the trace in mind,
we consider a hyperplane in $V$ which is not through zero, which is of the form
$$
L_y=\{x\in V:\lan x,y\ran=1\},
$$
for a nonzero $y\in W$. If $y\in C^\circ$ then it is clear that $C\cap L_y$ is nonempty.
We recall again that $y$ is an interior point of $C^\circ$ if and only if
$\lan x,y\ran>0$ for every nonzero $x\in C$. This is equivalent to the following geometric condition;
\begin{enumerate}
\item[$(C_1)$]
$L_y$ intersects every ray in $C$.
\end{enumerate}

\begin{proposition}\label{pro_prinSH1}
Suppose that $C$ is a closed convex cone in $V$ with nonempty interior
and $y\in W$ is an interior point of $C^\circ$. Then
\begin{equation}\label{corres}
H_z\mapsto H_z\cap L_y,\qquad z\in\partial C^\circ\setminus\{0\},
\end{equation}
is a one-to-one
correspondence between nontrivial supporting hyperplanes of $C$ and supporting hyperplanes of $C\cap L_y$.
Furthermore, $H_z$ is principal if and only if $H_z\cap L_y$ is principal.
\end{proposition}

\begin{proof}
Suppose that $H_z$ is a nontrivial supporting hyperplane of $C$, and take a nonzero $x\in C\cap H_z$.
Then $\lan x,y\ran>0=\lan x,z\ran$, and so $z$ is not a scalar multiple of $y$. Hence $H_z\cap L_y$ is a hyperplane
of the affine space $L_y$.
The ray generated by $x$ meets $L_y$ and so we see that $(C\cap L_y)\cap (H_z\cap L_y)$ is nonempty.
Since $C\subset H_z^+$ implies $C\cap L_y\subset H_z^+\cap L_y$, we see that $H_z\cap L_y$ is a supporting hyperplane of
$C\cap L_y$.

We proceed to show that $H_z$ coincides with the affine space $\Aff (H_z\cap L_y,0)$ in $V$ generated by $H_z\cap L_y$ and zero.
We note that $H_z\cap L_y$ is an affine space of codimension two in $V$ and $0\notin H_z\cap L_y$, and so
$\Aff(H_z\cap L_y,0)$ is a subspace of codimension one. We also note that $H_z$ is also a subspace of codimension one.
By the relation $H_z\supset \Aff(H_z\cap L_y,0)$, we conclude $H_z=\Aff(H_z\cap L_y,0)$.
Therefore, the correspondence (\ref{corres}) is one-to-one.

In order to show that the correspondence (\ref{corres}) is surjective, suppose that $K$ is
a supporting hyperplane of $C\cap L_y$. Then $K$ is an affine subspace of $V$ of codimension two, and it is of the form
$$
K=\{x\in V: \lan x,y\ran=1,\ \lan x,z\ran=0\}\subset L_y,
$$
for a nonzero $z\in W$, and so $H_z\cap L_y=K$. We choose the sign of $z$ so that
$$
C\cap L_y\subset \{x\in V: \lan x,y\ran=1,\ \lan x,z\ran\ge 0\}.
$$
If $x\in C$ is nonzero then we have $\alpha x\in C\cap L_y$ for some $\alpha>0$ by the property $(C_1)$, and so
$\lan \alpha x, z\ran\ge 0$. This implies $\lan x, z\ran\ge 0$, and so we have $z\in C^\circ$.
Since $K$ meets $C\cap L_y$, there is a nonzero $x\in C$ with $\lan x,z\ran=0$, and so $z\in\partial C^\circ$.
Therefore, $H_z$ is a nontrivial supporting hyperplane of $C$.
The last claim can be seen easily from Proposition \ref{pro_prinSH}.
\end{proof}

Suppose that closed convex sets $C$ and $D$ with the inclusion $C\supset D$ share a boundary point $x_0$. We say that
$C\supset D$ is {\sl locally distinguishable around $x_0$} if the set $N\cap (C\setminus D)$ is nonempty for every neighborhood $N$ of $x_0$.
If $C\supset D$ then every supporting hyperplane of $C$ through $x_0$ is also a supporting hyperplane of $D$ through $x_0$.
If there exists a supporting hyperplane $H$ of $D$ through $x_0$ which is not a supporting hyperplane of $C$,
it is clear that $C\supset D$ is locally distinguishable around $x_0$. See {\sc Figure 3}. In fact, we take
$x\in C$ satisfying $x\notin H^+$, then the open line segment from $x_0$ to $x$
lies in $C\setminus D$. If such a hyperplane $H$ can be chosen to be a principal supporting hyperplane of $D$,
then we say that $C\supset D$ is {\sl locally distinguishable around $x_0$ by principal supporting hyperplanes}.

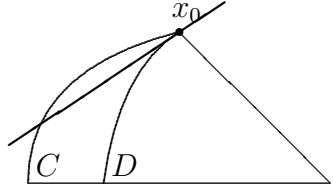
\begin{figure}
\begin{center}
\setlength{\unitlength}{1 truecm}
\begin{picture}(5,3)
\put(3,2){\circle*{0.1}}
\put(2.9,2.2){$x_0$}
\put(1.1,0.1){$C$}
\put(2.1,0.1){$D$}
\thinlines
\qbezier(1,0)(1,1.5)(3,2)
\qbezier(2,0)(2.2,1.5)(3,2)
\drawline(1,0)(5,0)
\drawline(3,2)(5,0)
\thicklines
\drawline(3,2)(0.75,0.5)
\drawline(3,2)(3.6,2.4)
\end{picture}
\end{center}
\caption{
When two convex sets $C$ and $D$ with $C\supsetneqq D$ have a common boundary point $x_0$
and there exists a supporting hyperplane of $D$ through $x_0$ which is not a supporting hyperplane of $C$,
then the inclusion $C\supsetneqq D$ is locally distinguishable around $x_0$.
}
\end{figure}

In the next section, we consider the case when $\blockpos_{k}$ and $\blockpos_{k+1}$ have a common boundary point $\varrho_{E^\perp}$.
We will look for principal supporting hyperplanes of $\blockpos_{k+1}$ which are not supporting hyperplanes of
$\blockpos_{k}$ to check if $\blockpos_{k}\supsetneqq\blockpos_{k+1}$ is locally distinguishable around $\varrho_{E^\perp}$.

%%%%%%%%%%%%%%%%%%%%%%%%%%%%%%%%%%%%%%%%%%%%%%%%%%
\section{Principal supporting hyperplanes for Schmidt number witnesses}

We denote by $\mathcal S_k$ the convex hull of pure states
onto vectors in $\mathbb C^m\ot \mathbb C^n$ with Schmidt rank $\le k$ in the affine subspace $\mathcal H$
of $(M_m\ot M_n)^{\rm h}$ with trace one. We will assume  $m\le n$ without loss of generality.
Then we get the chain
$$
{\mathcal S}_1\subsetneqq {\mathcal S}_2\subsetneqq\cdots \subsetneqq{\mathcal S}_k
\subsetneqq\cdots \subsetneqq {\mathcal S}_{m-1}\subsetneqq {\mathcal S}_m=\mathcal D
$$
of strict inclusions. A state in $\mathcal D$ is entangled if and only if it does not belong to ${\mathcal S}_1$.
A state is said to have {\sl Schmidt number $k$} when it belongs to ${\mathcal S}_k\setminus {\mathcal S}_{k-1}$.
The Schmidt number of a state $\varrho$ will be denoted by $\SN(\varrho)$.

A Hermitian matrix $W\in M_m\ot M_n$
is called {\sl $k$-blockpositive}  if
$\lan\xi|W|\xi\ran\ge 0$ for every $|\xi\ran$ with $\sr|\xi\ran\le k$.
It is clear that the dual cone $\mathcal S_k^\circ$ of $\mathcal S_k$ is just the convex cone
$\blockpos_k$ consisting of all $k$-blockpositive matrices.
Note that the hyperplane $\mathcal H$ in $(M_m\ot M_n)^{\rm h}$ determined by the trace
arises from the identity matrix, which is an interior point of $\blockpos_k^\circ=\mathcal S_k$.
Using Proposition \ref{pro_prinSH1}, we abuse the notation $\blockpos_k$ to denote both the \lq\lq convex cone\rq\rq\
and the normalized \lq\lq convex compact set\rq\rq\ of $k$-blockpositive matrices.
Then the dual cone of $\blockpos_k$ is also given by
the \lq\lq convex cone\rq\rq\ $\mathcal S_k$
by another abuse of notation. The convex cone $\mathcal S_k$ was introduced in \cite{eom-kye}
as the dual object of $k$-positive maps with respect to the bilinear pairing between mapping space and tensor product.

Now, we have the chain
$$
\mathcal S_1\subsetneqq \cdots\subsetneqq \mathcal S_k \subsetneqq \cdots\subsetneqq
\mathcal S_{m} =\mathcal D=\blockpos_{m}\subsetneqq\cdots\subsetneqq\blockpos_k \subsetneqq \cdots\subsetneqq\blockpos_1
$$
of strict inclusions of compact convex sets.
We note that $\SN(\varrho)\ge k$ if and only if $\varrho\notin \mathcal
S_{k-1}$ if and only if $\lan\varrho, W\ran<0$ for some $W\in\blockpos_{k-1}$,
where $\lan x,y\ran=\tr(xy)$ for Hermitian matrices $x$ and $y$. If $\SN(\varrho)=k$ then
such a $W$ belongs to $\blockpos_{k-1}\setminus\blockpos_{k}$. A
Hermitian matrix $W\in\blockpos_{k-1}\setminus\blockpos_{k}$ is
called a {\sl Schmidt number $k$ witness}.

By \cite{eom-kye}, a $k$-blockpositive matrix $W$ is an interior point of $\blockpos_k$ if and only if
$\lan\xi|W|\xi\ran>0$ for every nonzero $|\xi\ran$ with $\SR|\xi\ran\le k$.
Applying this to $W=\varrho_{E^\perp}$ and noting that $\lan\xi|\varrho_{E^\perp}|\xi\ran=0$ if and only if $|\xi\ran\in E$,
we see that
\begin{equation}\label{boundary}
\varrho_{E^\perp}\in\partial\blockpos_k\quad\Longleftrightarrow\quad S(E,k)\neq\{0\}\quad\Longleftrightarrow\quad \mu(E)\le k,
\end{equation}
for a nonzero proper subspace $E$ of $\mathbb C^m\ot\mathbb C^n$.

By the definition of $\mathcal S_k$, we see that every extreme ray of the convex cone $\mathcal S_k$ is
of the form $|\xi\ran\lan\xi|$ with $\SR|\xi\ran\le k$. Conversely, it is easily seen that
$|\xi\ran\lan\xi|$ generates an extreme ray of the cone $\mathcal S_k$ whenever $\SR|\xi\ran\le k$.
In fact, it is known \cite{{marcin_exp},{kye_ad_s}} that $|\xi\ran\lan\xi|$
generates an exposed extreme ray of the much bigger cone $\blockpos_1$,
as the Choi matrix of the map $a\mapsto s^*as:M_m\to M_n$ for an $m\times n$ matrix $s$.
By Propositions \ref {pro_prinSH} and \ref{pro_prinSH1}, we have the following:

\begin{theorem}\label{thm_corr}
Suppose that $E$ is a proper subspace of $\mathbb C^m\ot\mathbb C^n$.
%with $E^\perp\neq\{0\}$.
Then the principal supporting hyperplanes of
the convex set $\blockpos_k$
through $\varrho_{E^\perp}$ are exactly the hyperplanes $H_{|\xi\ran\lan\xi|}\cap\mathcal H$ for nonzero $|\xi\ran\in S(E,k)$,
and two such hyperplanes coincide if and only if the corresponding vectors are scalar multiples of each other.
\end{theorem}

Consequently, principal supporting hyperplanes of $\blockpos_k$ through $\varrho_{E^\perp}$ correspond to rays in $S(E,k)$,
and we expect that the inclusion $\blockpos_k\supsetneqq\blockpos_{k+1}$ is locally distinguishable
around $\varrho_{E^\perp}$ by principal supporting hyperplanes whenever $S^\prime(E,k+1)\neq\emptyset$.
In fact, we have the following:

\begin{proposition}\label{pro_detect}
Let $E$ be a nonzero proper subspace of $\mathbb C^m\ot\mathbb C^n$ and $k<m\meet n$, and suppose that $S(E,k)\neq\{0\}$.
Then the following are equivalent:
\begin{enumerate}
\item[{\rm (i)}]
$S^\prime(E,k+1)\neq\emptyset$,
\item[{\rm (ii)}]
the inclusion $\blockpos_k\supsetneqq\blockpos_{k+1}$ is locally distinguishable around $\varrho_{E^\perp}$ by principal supporting hyperplanes,
\item[{\rm (iii)}]
there exists a supporting hyperplane of $\blockpos_{k+1}$ through $\varrho_{E^\perp}$ which is not a supporting hyperplane of $\blockpos_k$.
\end{enumerate}
\end{proposition}

\begin{proof}
We first note that
$\varrho_{E^\perp}$ is a common boundary point of $\blockpos_\ell$ and $\blockpos_{\ell+1}$ by (\ref{boundary}).
The direction (i) $\Longrightarrow$ (ii) follows from
Proposition \ref{pro_prinSH1} and Theorem \ref{thm_corr}, and (ii) $\Rightarrow$ (iii) is clear.
In order to prove (iii) $\Rightarrow$ (i), it suffices to show that
$S^\prime(E,k+1)=\emptyset$ implies that  $\blockpos_k$ and $\blockpos_{k+1}$ have exactly the same supporting hyperplanes
through $\varrho_{E^\perp}$.
Let $K$ be a supporting hyperplane of $\blockpos_{k+1}$ through $\varrho_{E^\perp}$. By Proposition \ref{pro_prinSH1},
we have $K=H_z\cap\mathcal H$ for a nonzero $z\in\mathcal S_{k+1}$ with $\lan\varrho_{E^\perp},z\ran=0$.
We write $z=\sum_i|\xi_i\ran\lan\xi_i|$ with $\SR|\xi_i\ran\le k+1$. Then
$\sum_i\|P_{E^\perp}|\xi_i\ran\|^2=D\lan\varrho_{E^\perp},z\ran=0$, and so every $|\xi_i\ran$ belongs to $S(E,k+1)=S(E,k)$.
Hence we have $z\in\mathcal S_k=\blockpos_k^\circ$, and $K$ is a supporting hyperplane of $\blockpos_k$ through $\varrho_{E^\perp}$.
\end{proof}

The geometry of the algebraic variety $S(E,k)$
depends heavily on the dimension of $E$, as we see in Theorem \ref{main} below.
Recall that the set of all $d$ dimensional subspaces of $\mathbb C^m\ot\mathbb C^n$ is the
Grassmann manifold $\Gr(d,mn)$, which is an irreducible projective variety of dimension $d(mn-d)$.
We say that a property holds for {\sl generic} subspaces $E$ of dimension $d$ when it holds for every $E$
in a nonempty Zariski open subset of $\Gr(d,mn)$. Such a subset is dense in $\Gr(d,mn)$, and its complement has measure zero
with respect to the unitarily invariant probability measure. In particular, a randomly chosen subspace of dimension $d$
has a given generic property with probability one. Since $\Gr(d,mn)$ is irreducible, finitely many generic properties
hold simultaneously for generic subspaces.

\begin{theorem}\label{main} Let $m\le n$ and $k\in\{1,\dots,m\}$, and put $c_k=(m-k)(n-k)$.
For subspaces $E$ in $\mathbb C^m\ot\mathbb C^n$, we have the following:
\begin{enumerate}
\item[{\rm (i)}]
If $\dim E\le c_k$, then
$S(E,k)=\{0\}$
for generic $E$.
\item[{\rm (ii)}]
If $\dim E> c_k$, then $S(E,k)$ is nonzero and $$\dim S(E,k)\ge\dim E-c_k.$$
For generic $E$, we have $\dim S'(E,k)=\dim S(E,k)=\dim E-c_k$. In particular, $S^\prime(E,k)$ is nonempty.
\item[{\rm (iii)}]
If $\dim E= c_k+1$ and $E$ is generic, then
$S^\prime(E,k)=S(E,k)\setminus\{0\}$, and it
consists of finitely many rays whose cardinality is
$$
\displaystyle\prod_{i=0}^{m-k-1}\frac{(n+i)!\, i!}{(k+i)!\, (n-k+i)!},
$$
where a \emph{ray} means a $1$-dimensional subspace minus the origin,
that is, a nonzero vector up to scalar multiplication.
\end{enumerate}
\end{theorem}

\begin{proof}
To prove this theorem, we will use projective geometry, more precisely the dimension theory
in \cite[Chapters 11 and 12]{HarrisAG} as well as a well-known degree formula in \cite{Fulton}.

Recall that the projectivization of a complex vector space $U$ (respectively a cone $\Gamma$ in $U$) is defined as the quotient space
$$
\mathbb P U=(U\setminus\{0\})/\mathbb C^* \quad \text{(respectively}\quad \mathbb P \Gamma= (\Gamma\setminus\{0\})/\mathbb C^* \text{)},
$$
where $\mathbb C^*$ acts by scalar multiplication. Let
$$
C_k=\{|\xi\ran\in \mathbb C^m\ot \mathbb C^n: \SR |\xi\ran \le k\}
$$
so that $S(E,k)=C_k\cap E$, and $C_0=\{0\}$.
Note that $C_k$ is a cone, that is, closed under nonzero scalar multiplications, and it is closed
since it is the zero set of all $(k+1)\times(k+1)$ minors of the corresponding $m\times n$ matrices.

The dimension of a smooth variety is the same as its dimension as a complex manifold, i.e.
the number of coordinates in a local chart. For example, we have $\dim \mathbb P \mathbb C^N=N-1$
since $\frac{x_1}{x_N},\cdots, \frac{x_{N-1}}{x_N}$ are local coordinates at points with $x_N\ne 0$.
If an algebraic variety $X$ is not smooth, we can define the dimension of $X$
as the dimension of the smooth part $X^{\mathrm{sm}}$, the complement of singular points.
When a variety $X$ is a closed subset of a projective space $\mathbb P U$,
there are several equivalent definitions of the dimension of $X$ in terms of its intersection properties with linear subspaces.
See \cite[Chapter 11]{HarrisAG} for details.

By \cite[Proposition 12.2]{HarrisAG}, $\mathbb P C_k$ is irreducible and has codimension
$c_k$ in $\mathbb P (\mathbb C^m\ot \mathbb C^n).$
By the definition of $S(E,k)$, we have
$\mathbb P S(E,k)=\mathbb P E\cap \mathbb P C_k$
and
$$
\dim S(E,k)=\dim \left( \mathbb P E\cap \mathbb P C_k \right) +1
$$
when $S(E,k)$ is nonzero. Put $d=\dim E$, so that $\dim\mathbb P E=d-1$.

The statement (i) follows from \cite[Definition 11.2]{HarrisAG}, since $\dim\mathbb P E=d-1<c_k$.
The first part of (ii) follows by repeatedly applying \cite[Exercise 11.6]{HarrisAG}
since $\mathbb P E$ is an intersection of $mn-d$ hyperplanes. For generic $E$,
the subspace $\mathbb P E$ is the intersection of $mn-d$ general hyperplanes.
A general hyperplane contains no irreducible component of a given projective variety, and so
intersecting with it lowers the dimension of every irreducible component by one, by \cite[Exercise 11.6]{HarrisAG} again.
Therefore, for generic $E$, every irreducible component of $\mathbb P E\cap\mathbb P C_k$ has dimension exactly $d-1-c_k$.

We have $c_{k-1}-c_k=m+n-2k+1\ge 1$. Hence, for generic $E$, the set $\mathbb P E\cap\mathbb P C_{k-1}$ is either empty
or of dimension $d-1-c_{k-1}<d-1-c_k$. Therefore, it contains no irreducible component of $\mathbb P E\cap\mathbb P C_k$,
and $\mathbb P S^\prime(E,k)=(\mathbb P E\cap\mathbb P C_k)\setminus \mathbb P C_{k-1}$ is a dense open subset of
$\mathbb P E\cap\mathbb P C_k$. This shows $\dim S^\prime(E,k)=\dim S(E,k)=d-c_k$, and completes the proof of (ii).

Finally, if $d=c_k+1$ then $\mathbb P E\cap\mathbb P C_k$ is a nonempty finite set for generic $E$ by the above argument, and
$\mathbb P E\cap\mathbb P C_{k-1}=\emptyset$ for generic $E$ by (i), since $d=c_k+1\le c_{k-1}$.
Therefore, we have $S^\prime(E,k)=S(E,k)\setminus\{0\}$.
The degree of $\mathbb P C_k$ is the cardinality of $\mathbb P C_k\cap \mathbb P E$
for generic $E$ by \cite[Definition 18.1 (iii)]{HarrisAG}, and it is given by the formula in (iii) by \cite[Example 14.4.14]{Fulton}.
\end{proof}

Note that the Schmidt rank of a vector in $\mathbb C^m\ot\mathbb C^n$  is nothing but the rank of the corresponding
$m\times n$ matrix, and the property of \lq\lq rank $\le k$\rq\rq\ of an $m\times n$ matrix is locally determined by a system
of $(m-k)(n-k)$ equations. On the other hand, a subspace of dimension $d$ has
$d-1$ unknowns up to scalar multiplications. Comparing the numbers of equations and unknowns,
the claims (i) and (ii) of Theorem \ref{main} may look reasonable.
The claim (iii) deals with
the case when the  numbers of equations and unknowns coincide.

Recall that the number $\kappa_d$ is given by (\ref{def_kappa}), that is, $c_{\kappa_d}<d\le c_{\kappa_d-1}$.
As an immediate consequence of Theorem \ref{main}, we have the following:
\begin{itemize}
\item
for generic subspaces $E$, we have $S(E,\ell)=\{0\}$ for $\ell=1,\dots,\kappa_{\dim E}-1$;
\item
for generic subspaces $E$, we have $S^\prime(E,\ell)\neq \emptyset$ for $\ell=\kappa_{\dim E},\dots,m$.
\end{itemize}
In particular, $\mu(E)=\kappa_{\dim E}$ for generic $E$ as in \cite{CMW08}.
Therefore, for generic subspaces $E$,
there exist Schmidt number $\ell$ witnesses outside of the face $F_{E^\perp}$
for $\ell=2,3,\dots,\kappa_{\dim E}$ by \cite{han_kye_2025_a}, and Proposition \ref{pro_detect} gives the following:

\begin{theorem}\label{main_snw}
Let $E$ be a generic subspace with $\kappa=\kappa_{\dim E}$ and $E\neq\mathbb C^m\ot\mathbb C^n$. Then the inclusion
$\blockpos_{\ell}\supsetneqq\blockpos_{\ell+1}$ is locally distinguishable around $\varrho_{E^\perp}$ by principal supporting
hyperplanes for $\ell=\kappa,\kappa+1,\dots,m\meet n-1$. \end{theorem}

As a consequence, every neighborhood of $\varrho_{E^\perp}$
contains Schmidt number $\ell$ witnesses for $\ell=\kappa+1,\kappa+2,\dots, m\meet n$.
We will see in Section \ref{sec_general} that this holds
for every nonzero proper subspace $E$ with $\kappa$ replaced by $\mu(E)$.
What genericity provides is that the inclusions are already detected by principal supporting hyperplanes.

%%%%%%%%%%%%%%%%%%%%%%%%%%%%%%%%%%%%%%%%%%%%%%%%%%
\section{Examples}

For a non-generic case, the inclusion $\blockpos_{\ell}\supsetneqq\blockpos_{\ell+1}$ may not be distinguished
by principal supporting hyperplanes. For example, we consider a one dimensional space $E$ spanned by
a vector $|\xi\ran$ of Schmidt rank $k$, with $k< m\meet n$. In this case, we have
$S^\prime (E,\ell)=\emptyset$ for $\ell >k$, and so
$\blockpos_{\ell-1}\supsetneqq\blockpos_\ell$ cannot be locally distinguished
by principal supporting hyperplanes around $\varrho_{E^\perp}$ for $\ell=k+1,\dots,m\meet n$.
In fact, Proposition \ref{pro_detect} shows that these inclusions are not detected by any supporting hyperplane
through $\varrho_{E^\perp}$, although they are locally distinguishable by Theorem \ref{thm_general} below.
We provide concrete examples.

\medskip
\noindent
{\bf Example 1}.
Take the one dimensional subspace $E_1$ (respectively $E_2$) of $\mathbb C^2\ot\mathbb C^2$ spanned by $|00\ran$
(respectively $|00\ran+|11\ran$)  which has Schmidt rank one (respectively two). Considering the two
dimensional section determined by the projection states $\varrho_{E_1}$, $\varrho_{E_2}$ and the
maximally mixed state $\varrho_*$, we define two parameter family
$$
\begin{aligned}
X_{s,t}
&=(1-s-t)\varrho_*+s \varrho_{E_1} +t \varrho_{E_2}\\
&=\frac 14\left(\begin{matrix}
1+3s+t &\cdot &\cdot &2t\\
\cdot &1-s-t &\cdot &\cdot\\
\cdot &\cdot &1-s-t &\cdot\\
2t&\cdot &\cdot & 1-s+t
\end{matrix}\right)\in M_2\ot M_2.
\end{aligned}
$$
Then $X_{s,t}\in\blockpos_2=\mathcal D$ if and only if $(s,t)$ satisfies
$$
\max\{s-1,-3s-1\}\le t \le -s+1,\qquad
3s^2-2st+3t^2-2s-2t-1\le 0.
$$
On the other hand, we use Theorem 5.5 in \cite{han_kye_multi} to see that $X_{s,t}\in\blockpos_1$ if and only if
$$
\begin{aligned}
&\max\{s-1,-3s-1\}\le t \le -s+1,\\
&\sqrt{(1+3s+t)(1-s+t)}+(1-s-t)\ge |2t|.
\end{aligned}
$$
See {\sc Figure 4}.

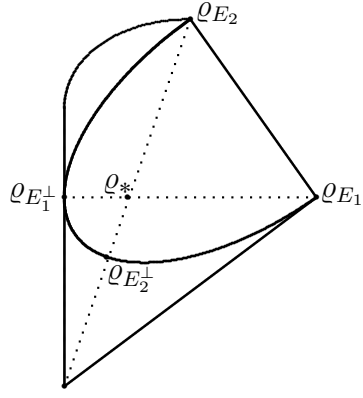
\begin{figure}
\begin{center}
\setlength{\unitlength}{2.5 truecm}
\begin{picture}(1.5,2)
\thicklines

\put(0.4,1){\circle*{0.03}}\put(0.27,1.05){$\varrho_{*}$}
\put(0.06667,1){\circle*{0.03}}\put(-0.23,1){$\varrho_{E_1^\perp}$}
\put(0.733333,1.942809){\circle*{0.03}}\put(0.76,1.96){$\varrho_{E_2}$}
\put(0.28889,0.68574){\circle*{0.03}}\put(0.28889,0.57){$\varrho_{E_2^\perp}$}
\put(1.399999,1){\circle*{0.03}}\put(1.419999,1){$\varrho_{E_1}$}
\put(0.06667,0.05719){\circle*{0.03}}
\dottedline{0.05}(0.06667,1)(1.399999,1)
\dottedline{0.05}(0.733333,1.942809)(0.06667,0.05719)

\dottedline{0.005}(0.733333,1.942809)(1.399999,1)(0.06667,0.05719)(0.06667,1.471404)

\put(0.06667,1){\circle*{0.001}}
\put(0.06669,1.004194){\circle*{0.001}}
\put(0.06675,1.008398){\circle*{0.001}}
\put(0.06685,1.012611){\circle*{0.001}}
\put(0.06698,1.016831){\circle*{0.001}}
\put(0.06716,1.02106){\circle*{0.001}}
\put(0.06737,1.025296){\circle*{0.001}}
\put(0.06762,1.02954){\circle*{0.001}}
\put(0.06791,1.033791){\circle*{0.001}}
\put(0.06824,1.038048){\circle*{0.001}}
\put(0.0686,1.042312){\circle*{0.001}}
\put(0.069,1.046581){\circle*{0.001}}
\put(0.06944,1.050856){\circle*{0.001}}
\put(0.06991,1.055137){\circle*{0.001}}
\put(0.07042,1.059423){\circle*{0.001}}
\put(0.07096,1.063713){\circle*{0.001}}
\put(0.07154,1.068008){\circle*{0.001}}
\put(0.07216,1.072307){\circle*{0.001}}
\put(0.0728,1.07661){\circle*{0.001}}
\put(0.07349,1.080916){\circle*{0.001}}
\put(0.0742,1.085225){\circle*{0.001}}
\put(0.07496,1.089538){\circle*{0.001}}
\put(0.07574,1.093852){\circle*{0.001}}
\put(0.07656,1.09817){\circle*{0.001}}
\put(0.07741,1.102489){\circle*{0.001}}
\put(0.07829,1.10681){\circle*{0.001}}
\put(0.07921,1.111133){\circle*{0.001}}
\put(0.08015,1.115456){\circle*{0.001}}
\put(0.08113,1.119781){\circle*{0.001}}
\put(0.08214,1.124106){\circle*{0.001}}
\put(0.08318,1.128432){\circle*{0.001}}
\put(0.08426,1.132757){\circle*{0.001}}
\put(0.08536,1.137083){\circle*{0.001}}
\put(0.08649,1.141408){\circle*{0.001}}
\put(0.08765,1.145732){\circle*{0.001}}
\put(0.08884,1.150055){\circle*{0.001}}
\put(0.09006,1.154378){\circle*{0.001}}
\put(0.09131,1.158698){\circle*{0.001}}
\put(0.09259,1.163017){\circle*{0.001}}
\put(0.0939,1.167334){\circle*{0.001}}
\put(0.09523,1.171649){\circle*{0.001}}
\put(0.09659,1.175961){\circle*{0.001}}
\put(0.09798,1.180271){\circle*{0.001}}
\put(0.0994,1.184577){\circle*{0.001}}
\put(0.10084,1.188881){\circle*{0.001}}
\put(0.10231,1.193181){\circle*{0.001}}
\put(0.1038,1.197477){\circle*{0.001}}
\put(0.10532,1.20177){\circle*{0.001}}
\put(0.10687,1.206059){\circle*{0.001}}
\put(0.10844,1.210343){\circle*{0.001}}
\put(0.11003,1.214623){\circle*{0.001}}
\put(0.11165,1.218898){\circle*{0.001}}
\put(0.1133,1.223168){\circle*{0.001}}
\put(0.11496,1.227433){\circle*{0.001}}
\put(0.11665,1.231693){\circle*{0.001}}
\put(0.11837,1.235947){\circle*{0.001}}
\put(0.12011,1.240196){\circle*{0.001}}
\put(0.12187,1.244438){\circle*{0.001}}
\put(0.12365,1.248675){\circle*{0.001}}
\put(0.12545,1.252905){\circle*{0.001}}
\put(0.12728,1.257129){\circle*{0.001}}
\put(0.12912,1.261346){\circle*{0.001}}
\put(0.13099,1.265557){\circle*{0.001}}
\put(0.13288,1.26976){\circle*{0.001}}
\put(0.13479,1.273956){\circle*{0.001}}
\put(0.13672,1.278145){\circle*{0.001}}
\put(0.13867,1.282327){\circle*{0.001}}
\put(0.14064,1.286501){\circle*{0.001}}
\put(0.14263,1.290666){\circle*{0.001}}
\put(0.14464,1.294824){\circle*{0.001}}
\put(0.14666,1.298974){\circle*{0.001}}
\put(0.14871,1.303116){\circle*{0.001}}
\put(0.15077,1.307249){\circle*{0.001}}
\put(0.15285,1.311373){\circle*{0.001}}
\put(0.15495,1.315489){\circle*{0.001}}
\put(0.15707,1.319596){\circle*{0.001}}
\put(0.1592,1.323694){\circle*{0.001}}
\put(0.16135,1.327782){\circle*{0.001}}
\put(0.16352,1.331862){\circle*{0.001}}
\put(0.1657,1.335932){\circle*{0.001}}
\put(0.1679,1.339992){\circle*{0.001}}
\put(0.17011,1.344043){\circle*{0.001}}
\put(0.17234,1.348084){\circle*{0.001}}
\put(0.17459,1.352115){\circle*{0.001}}
\put(0.17685,1.356136){\circle*{0.001}}
\put(0.17912,1.360147){\circle*{0.001}}
\put(0.18141,1.364148){\circle*{0.001}}
\put(0.18371,1.368139){\circle*{0.001}}
\put(0.18603,1.372119){\circle*{0.001}}
\put(0.18836,1.376088){\circle*{0.001}}
\put(0.1907,1.380047){\circle*{0.001}}
\put(0.19306,1.383994){\circle*{0.001}}
\put(0.19543,1.387932){\circle*{0.001}}
\put(0.19781,1.391858){\circle*{0.001}}
\put(0.20021,1.395773){\circle*{0.001}}
\put(0.20261,1.399677){\circle*{0.001}}
\put(0.20503,1.403569){\circle*{0.001}}
\put(0.20746,1.40745){\circle*{0.001}}
\put(0.2099,1.41132){\circle*{0.001}}
\put(0.21236,1.415179){\circle*{0.001}}
\put(0.21482,1.419026){\circle*{0.001}}
\put(0.21729,1.422861){\circle*{0.001}}
\put(0.21978,1.426685){\circle*{0.001}}
\put(0.22227,1.430496){\circle*{0.001}}
\put(0.22478,1.434296){\circle*{0.001}}
\put(0.22729,1.438084){\circle*{0.001}}
\put(0.22982,1.44186){\circle*{0.001}}
\put(0.23235,1.445624){\circle*{0.001}}
\put(0.2349,1.449376){\circle*{0.001}}
\put(0.23745,1.453115){\circle*{0.001}}
\put(0.24001,1.456843){\circle*{0.001}}
\put(0.24258,1.460558){\circle*{0.001}}
\put(0.24516,1.464261){\circle*{0.001}}
\put(0.24774,1.467951){\circle*{0.001}}
\put(0.25033,1.471629){\circle*{0.001}}
\put(0.25294,1.475294){\circle*{0.001}}
\put(0.25554,1.478947){\circle*{0.001}}
\put(0.25816,1.482587){\circle*{0.001}}
\put(0.26078,1.486215){\circle*{0.001}}
\put(0.26341,1.48983){\circle*{0.001}}
\put(0.26605,1.493432){\circle*{0.001}}
\put(0.26869,1.497022){\circle*{0.001}}
\put(0.27134,1.500598){\circle*{0.001}}
\put(0.274,1.504162){\circle*{0.001}}
\put(0.27666,1.507713){\circle*{0.001}}
\put(0.27932,1.511252){\circle*{0.001}}
\put(0.282,1.514777){\circle*{0.001}}
\put(0.28468,1.518289){\circle*{0.001}}
\put(0.28736,1.521789){\circle*{0.001}}
\put(0.29005,1.525275){\circle*{0.001}}
\put(0.29274,1.528748){\circle*{0.001}}
\put(0.29544,1.532208){\circle*{0.001}}
\put(0.29814,1.535656){\circle*{0.001}}
\put(0.30085,1.53909){\circle*{0.001}}
\put(0.30356,1.542511){\circle*{0.001}}
\put(0.30627,1.545919){\circle*{0.001}}
\put(0.30899,1.549314){\circle*{0.001}}
\put(0.31171,1.552695){\circle*{0.001}}
\put(0.31444,1.556064){\circle*{0.001}}
\put(0.31717,1.559419){\circle*{0.001}}
\put(0.3199,1.562761){\circle*{0.001}}
\put(0.32264,1.566091){\circle*{0.001}}
\put(0.32538,1.569406){\circle*{0.001}}
\put(0.32812,1.572709){\circle*{0.001}}
\put(0.33086,1.575998){\circle*{0.001}}
\put(0.33361,1.579275){\circle*{0.001}}
\put(0.33636,1.582538){\circle*{0.001}}
\put(0.33911,1.585788){\circle*{0.001}}
\put(0.34186,1.589024){\circle*{0.001}}
\put(0.34462,1.592248){\circle*{0.001}}
\put(0.34737,1.595458){\circle*{0.001}}
\put(0.35013,1.598655){\circle*{0.001}}
\put(0.35289,1.601839){\circle*{0.001}}
\put(0.35565,1.605009){\circle*{0.001}}
\put(0.35842,1.608167){\circle*{0.001}}
\put(0.36118,1.611311){\circle*{0.001}}
\put(0.36395,1.614442){\circle*{0.001}}
\put(0.36671,1.61756){\circle*{0.001}}
\put(0.36948,1.620665){\circle*{0.001}}
\put(0.37224,1.623757){\circle*{0.001}}
\put(0.37501,1.626836){\circle*{0.001}}
\put(0.37778,1.629901){\circle*{0.001}}
\put(0.38055,1.632953){\circle*{0.001}}
\put(0.38332,1.635993){\circle*{0.001}}
\put(0.38609,1.639019){\circle*{0.001}}
\put(0.38886,1.642032){\circle*{0.001}}
\put(0.39162,1.645032){\circle*{0.001}}
\put(0.39439,1.64802){\circle*{0.001}}
\put(0.39716,1.650994){\circle*{0.001}}
\put(0.39993,1.653955){\circle*{0.001}}
\put(0.402688,1.656903){\circle*{0.001}}
\put(0.405454,1.659839){\circle*{0.001}}
\put(0.40822,1.662761){\circle*{0.001}}
\put(0.410984,1.665671){\circle*{0.001}}
\put(0.413748,1.668567){\circle*{0.001}}
\put(0.41651,1.671451){\circle*{0.001}}
\put(0.419271,1.674322){\circle*{0.001}}
\put(0.422031,1.67718){\circle*{0.001}}
\put(0.42479,1.680026){\circle*{0.001}}
\put(0.427547,1.682859){\circle*{0.001}}
\put(0.430303,1.685679){\circle*{0.001}}
\put(0.433057,1.688486){\circle*{0.001}}
\put(0.435809,1.691281){\circle*{0.001}}
\put(0.438559,1.694063){\circle*{0.001}}
\put(0.441308,1.696833){\circle*{0.001}}
\put(0.444055,1.699589){\circle*{0.001}}
\put(0.446799,1.702334){\circle*{0.001}}
\put(0.449542,1.705066){\circle*{0.001}}
\put(0.452282,1.707785){\circle*{0.001}}
\put(0.45502,1.710492){\circle*{0.001}}
\put(0.457756,1.713187){\circle*{0.001}}
\put(0.460489,1.715869){\circle*{0.001}}
\put(0.46322,1.718539){\circle*{0.001}}
\put(0.465948,1.721196){\circle*{0.001}}
\put(0.468673,1.723842){\circle*{0.001}}
\put(0.471396,1.726475){\circle*{0.001}}
\put(0.474116,1.729096){\circle*{0.001}}
\put(0.476833,1.731704){\circle*{0.001}}
\put(0.479547,1.734301){\circle*{0.001}}
\put(0.482259,1.736885){\circle*{0.001}}
\put(0.484967,1.739458){\circle*{0.001}}
\put(0.487672,1.742018){\circle*{0.001}}
\put(0.490374,1.744566){\circle*{0.001}}
\put(0.493072,1.747103){\circle*{0.001}}
\put(0.495768,1.749627){\circle*{0.001}}
\put(0.49846,1.75214){\circle*{0.001}}
\put(0.501149,1.754641){\circle*{0.001}}
\put(0.503834,1.75713){\circle*{0.001}}
\put(0.506515,1.759607){\circle*{0.001}}
\put(0.509193,1.762072){\circle*{0.001}}
\put(0.511868,1.764526){\circle*{0.001}}
\put(0.514538,1.766968){\circle*{0.001}}
\put(0.517205,1.769399){\circle*{0.001}}
\put(0.519869,1.771818){\circle*{0.001}}
\put(0.522528,1.774225){\circle*{0.001}}
\put(0.525184,1.776621){\circle*{0.001}}
\put(0.527835,1.779006){\circle*{0.001}}
\put(0.530483,1.781379){\circle*{0.001}}
\put(0.533126,1.783741){\circle*{0.001}}
\put(0.535766,1.786091){\circle*{0.001}}
\put(0.538401,1.78843){\circle*{0.001}}
\put(0.541032,1.790758){\circle*{0.001}}
\put(0.543659,1.793075){\circle*{0.001}}
\put(0.546282,1.795381){\circle*{0.001}}
\put(0.5489,1.797675){\circle*{0.001}}
\put(0.551515,1.799959){\circle*{0.001}}
\put(0.554124,1.802231){\circle*{0.001}}
\put(0.55673,1.804492){\circle*{0.001}}
\put(0.559331,1.806743){\circle*{0.001}}
\put(0.561927,1.808983){\circle*{0.001}}
\put(0.564519,1.811211){\circle*{0.001}}
\put(0.567107,1.813429){\circle*{0.001}}
\put(0.56969,1.815636){\circle*{0.001}}
\put(0.572268,1.817833){\circle*{0.001}}
\put(0.574841,1.820018){\circle*{0.001}}
\put(0.57741,1.822193){\circle*{0.001}}
\put(0.579974,1.824358){\circle*{0.001}}
\put(0.582534,1.826512){\circle*{0.001}}
\put(0.585089,1.828655){\circle*{0.001}}
\put(0.587638,1.830788){\circle*{0.001}}
\put(0.590184,1.832911){\circle*{0.001}}
\put(0.592724,1.835023){\circle*{0.001}}
\put(0.595259,1.837124){\circle*{0.001}}
\put(0.59779,1.839216){\circle*{0.001}}
\put(0.600315,1.841297){\circle*{0.001}}
\put(0.602836,1.843368){\circle*{0.001}}
\put(0.605351,1.845429){\circle*{0.001}}
\put(0.607862,1.84748){\circle*{0.001}}
\put(0.610367,1.84952){\circle*{0.001}}
\put(0.612868,1.851551){\circle*{0.001}}
\put(0.615363,1.853571){\circle*{0.001}}
\put(0.617853,1.855582){\circle*{0.001}}
\put(0.620338,1.857583){\circle*{0.001}}
\put(0.622819,1.859574){\circle*{0.001}}
\put(0.625293,1.861555){\circle*{0.001}}
\put(0.627763,1.863526){\circle*{0.001}}
\put(0.630228,1.865488){\circle*{0.001}}
\put(0.632687,1.86744){\circle*{0.001}}
\put(0.635141,1.869382){\circle*{0.001}}
\put(0.63759,1.871314){\circle*{0.001}}
\put(0.640033,1.873237){\circle*{0.001}}
\put(0.642472,1.875151){\circle*{0.001}}
\put(0.644905,1.877055){\circle*{0.001}}
\put(0.647332,1.87895){\circle*{0.001}}
\put(0.649755,1.880835){\circle*{0.001}}
\put(0.652172,1.882711){\circle*{0.001}}
\put(0.654584,1.884577){\circle*{0.001}}
\put(0.65699,1.886435){\circle*{0.001}}
\put(0.659391,1.888283){\circle*{0.001}}
\put(0.661786,1.890122){\circle*{0.001}}
\put(0.664177,1.891952){\circle*{0.001}}
\put(0.666561,1.893772){\circle*{0.001}}
\put(0.668941,1.895584){\circle*{0.001}}
\put(0.671315,1.897387){\circle*{0.001}}
\put(0.673683,1.89918){\circle*{0.001}}
\put(0.676046,1.900965){\circle*{0.001}}
\put(0.678404,1.902741){\circle*{0.001}}
\put(0.680756,1.904508){\circle*{0.001}}
\put(0.683102,1.906267){\circle*{0.001}}
\put(0.685444,1.908016){\circle*{0.001}}
\put(0.687779,1.909757){\circle*{0.001}}
\put(0.690109,1.911489){\circle*{0.001}}
\put(0.692434,1.913213){\circle*{0.001}}
\put(0.694753,1.914928){\circle*{0.001}}
\put(0.697067,1.916634){\circle*{0.001}}
\put(0.699375,1.918332){\circle*{0.001}}
\put(0.701678,1.920022){\circle*{0.001}}
\put(0.703975,1.921703){\circle*{0.001}}
\put(0.706266,1.923375){\circle*{0.001}}
\put(0.708552,1.92504){\circle*{0.001}}
\put(0.710833,1.926696){\circle*{0.001}}
\put(0.713108,1.928343){\circle*{0.001}}
\put(0.715377,1.929983){\circle*{0.001}}
\put(0.717641,1.931614){\circle*{0.001}}
\put(0.7199,1.933238){\circle*{0.001}}
\put(0.722152,1.934853){\circle*{0.001}}
\put(0.7244,1.93646){\circle*{0.001}}
\put(0.726641,1.938059){\circle*{0.001}}
\put(0.728877,1.93965){\circle*{0.001}}
\put(0.731108,1.941233){\circle*{0.001}}
\put(0.733333,1.942809){\circle*{0.001}}
\put(0.06669,0.99582){\circle*{0.001}}
\put(0.06675,0.99164){\circle*{0.001}}
\put(0.06685,0.98748){\circle*{0.001}}
\put(0.06699,0.98332){\circle*{0.001}}
\put(0.06717,0.97918){\circle*{0.001}}
\put(0.06739,0.97504){\circle*{0.001}}
\put(0.06765,0.97092){\circle*{0.001}}
\put(0.06796,0.96681){\circle*{0.001}}
\put(0.0683,0.96271){\circle*{0.001}}
\put(0.06869,0.95863){\circle*{0.001}}
\put(0.06912,0.95455){\circle*{0.001}}
\put(0.06959,0.95049){\circle*{0.001}}
\put(0.0701,0.94645){\circle*{0.001}}
\put(0.07066,0.94241){\circle*{0.001}}
\put(0.07126,0.93839){\circle*{0.001}}
\put(0.0719,0.93439){\circle*{0.001}}
\put(0.07258,0.9304){\circle*{0.001}}
\put(0.07331,0.92643){\circle*{0.001}}
\put(0.07409,0.92247){\circle*{0.001}}
\put(0.0749,0.91853){\circle*{0.001}}
\put(0.07576,0.9146){\circle*{0.001}}
\put(0.07667,0.91069){\circle*{0.001}}
\put(0.07762,0.9068){\circle*{0.001}}
\put(0.07861,0.90293){\circle*{0.001}}
\put(0.07965,0.89907){\circle*{0.001}}
\put(0.08073,0.89523){\circle*{0.001}}
\put(0.08186,0.89141){\circle*{0.001}}
\put(0.08303,0.88761){\circle*{0.001}}
\put(0.08425,0.88383){\circle*{0.001}}
\put(0.08552,0.88007){\circle*{0.001}}
\put(0.08683,0.87633){\circle*{0.001}}
\put(0.08818,0.87261){\circle*{0.001}}
\put(0.08959,0.86891){\circle*{0.001}}
\put(0.09103,0.86523){\circle*{0.001}}
\put(0.09253,0.86157){\circle*{0.001}}
\put(0.09407,0.85794){\circle*{0.001}}
\put(0.09566,0.85432){\circle*{0.001}}
\put(0.09729,0.85073){\circle*{0.001}}
\put(0.09897,0.84717){\circle*{0.001}}
\put(0.1007,0.84362){\circle*{0.001}}
\put(0.10247,0.8401){\circle*{0.001}}
\put(0.10429,0.8366){\circle*{0.001}}
\put(0.10616,0.83313){\circle*{0.001}}
\put(0.10807,0.82968){\circle*{0.001}}
\put(0.11003,0.82626){\circle*{0.001}}
\put(0.11204,0.82287){\circle*{0.001}}
\put(0.1141,0.81949){\circle*{0.001}}
\put(0.1162,0.81615){\circle*{0.001}}
\put(0.11835,0.81283){\circle*{0.001}}
\put(0.12054,0.80954){\circle*{0.001}}
\put(0.12279,0.80627){\circle*{0.001}}
\put(0.12508,0.80304){\circle*{0.001}}
\put(0.12742,0.79983){\circle*{0.001}}
\put(0.1298,0.79665){\circle*{0.001}}
\put(0.13223,0.79349){\circle*{0.001}}
\put(0.13471,0.79037){\circle*{0.001}}
\put(0.13724,0.78727){\circle*{0.001}}
\put(0.13981,0.78421){\circle*{0.001}}
\put(0.14244,0.78117){\circle*{0.001}}
\put(0.1451,0.77817){\circle*{0.001}}
\put(0.14782,0.77519){\circle*{0.001}}
\put(0.15058,0.77225){\circle*{0.001}}
\put(0.15339,0.76933){\circle*{0.001}}
\put(0.15624,0.76645){\circle*{0.001}}
\put(0.15914,0.7636){\circle*{0.001}}
\put(0.16209,0.76078){\circle*{0.001}}
\put(0.16509,0.758){\circle*{0.001}}
\put(0.16813,0.75524){\circle*{0.001}}
\put(0.17122,0.75252){\circle*{0.001}}
\put(0.17435,0.74983){\circle*{0.001}}
\put(0.17753,0.74718){\circle*{0.001}}
\put(0.18075,0.74456){\circle*{0.001}}
\put(0.18402,0.74197){\circle*{0.001}}
\put(0.18734,0.73941){\circle*{0.001}}
\put(0.1907,0.7369){\circle*{0.001}}
\put(0.19411,0.73441){\circle*{0.001}}
\put(0.19756,0.73196){\circle*{0.001}}
\put(0.20106,0.72955){\circle*{0.001}}
\put(0.2046,0.72717){\circle*{0.001}}
\put(0.20819,0.72482){\circle*{0.001}}
\put(0.21182,0.72252){\circle*{0.001}}
\put(0.21549,0.72024){\circle*{0.001}}
\put(0.21921,0.71801){\circle*{0.001}}
\put(0.22297,0.71581){\circle*{0.001}}
\put(0.22678,0.71365){\circle*{0.001}}
\put(0.23063,0.71152){\circle*{0.001}}
\put(0.23452,0.70943){\circle*{0.001}}
\put(0.23845,0.70738){\circle*{0.001}}
\put(0.24243,0.70537){\circle*{0.001}}
\put(0.24645,0.70339){\circle*{0.001}}
\put(0.25051,0.70145){\circle*{0.001}}
\put(0.25461,0.69955){\circle*{0.001}}
\put(0.25875,0.69769){\circle*{0.001}}
\put(0.26294,0.69587){\circle*{0.001}}
\put(0.26716,0.69408){\circle*{0.001}}
\put(0.27143,0.69233){\circle*{0.001}}
\put(0.27574,0.69063){\circle*{0.001}}
\put(0.28008,0.68896){\circle*{0.001}}
\put(0.28447,0.68733){\circle*{0.001}}
\put(0.28889,0.68574){\circle*{0.001}} \put(0.4,1){\circle*{0.001}}
\put(0.29336,0.68418){\circle*{0.001}}
\put(0.29786,0.68267){\circle*{0.001}}
\put(0.3024,0.6812){\circle*{0.001}}
\put(0.30698,0.67977){\circle*{0.001}}
\put(0.3116,0.67837){\circle*{0.001}}
\put(0.31625,0.67702){\circle*{0.001}}
\put(0.32094,0.6757){\circle*{0.001}}
\put(0.32566,0.67443){\circle*{0.001}}
\put(0.33043,0.6732){\circle*{0.001}}
\put(0.33522,0.672){\circle*{0.001}}
\put(0.34006,0.67085){\circle*{0.001}}
\put(0.34493,0.66974){\circle*{0.001}}
\put(0.34983,0.66866){\circle*{0.001}}
\put(0.35477,0.66763){\circle*{0.001}}
\put(0.35974,0.66664){\circle*{0.001}}
\put(0.36474,0.66569){\circle*{0.001}}
\put(0.36978,0.66477){\circle*{0.001}}
\put(0.37485,0.6639){\circle*{0.001}}
\put(0.37995,0.66307){\circle*{0.001}}
\put(0.38508,0.66228){\circle*{0.001}}
\put(0.39024,0.66153){\circle*{0.001}}
\put(0.39544,0.66082){\circle*{0.001}}
\put(0.400658,0.66015){\circle*{0.001}}
\put(0.405912,0.65953){\circle*{0.001}}
\put(0.411197,0.65894){\circle*{0.001}}
\put(0.41651,0.65839){\circle*{0.001}}
\put(0.421852,0.65788){\circle*{0.001}}
\put(0.427222,0.65741){\circle*{0.001}}
\put(0.432619,0.65699){\circle*{0.001}}
\put(0.438044,0.6566){\circle*{0.001}}
\put(0.443495,0.65625){\circle*{0.001}}
\put(0.448973,0.65595){\circle*{0.001}}
\put(0.454476,0.65568){\circle*{0.001}}
\put(0.460005,0.65545){\circle*{0.001}}
\put(0.465558,0.65527){\circle*{0.001}}
\put(0.471136,0.65512){\circle*{0.001}}
\put(0.476738,0.65501){\circle*{0.001}}
\put(0.482363,0.65494){\circle*{0.001}}
\put(0.488011,0.65491){\circle*{0.001}}
\put(0.493681,0.65492){\circle*{0.001}}
\put(0.499374,0.65497){\circle*{0.001}}
\put(0.505088,0.65506){\circle*{0.001}}
\put(0.510823,0.65519){\circle*{0.001}}
\put(0.516578,0.65536){\circle*{0.001}}
\put(0.522354,0.65556){\circle*{0.001}}
\put(0.528149,0.65581){\circle*{0.001}}
\put(0.533963,0.65609){\circle*{0.001}}
\put(0.539796,0.65641){\circle*{0.001}}
\put(0.545646,0.65677){\circle*{0.001}}
\put(0.551515,0.65717){\circle*{0.001}}
\put(0.5574,0.6576){\circle*{0.001}}
\put(0.563302,0.65807){\circle*{0.001}}
\put(0.56922,0.65858){\circle*{0.001}}
\put(0.575154,0.65913){\circle*{0.001}}
\put(0.581103,0.65971){\circle*{0.001}}
\put(0.587066,0.66033){\circle*{0.001}}
\put(0.593044,0.66099){\circle*{0.001}}
\put(0.599036,0.66168){\circle*{0.001}}
\put(0.60504,0.66241){\circle*{0.001}}
\put(0.611057,0.66318){\circle*{0.001}}
\put(0.617087,0.66398){\circle*{0.001}}
\put(0.623128,0.66482){\circle*{0.001}}
\put(0.629181,0.66569){\circle*{0.001}}
\put(0.635244,0.6666){\circle*{0.001}}
\put(0.641317,0.66755){\circle*{0.001}}
\put(0.647401,0.66852){\circle*{0.001}}
\put(0.653493,0.66954){\circle*{0.001}}
\put(0.659595,0.67058){\circle*{0.001}}
\put(0.665705,0.67166){\circle*{0.001}}
\put(0.671823,0.67278){\circle*{0.001}}
\put(0.677948,0.67393){\circle*{0.001}}
\put(0.68408,0.67511){\circle*{0.001}}
\put(0.690219,0.67633){\circle*{0.001}}
\put(0.696363,0.67757){\circle*{0.001}}
\put(0.702514,0.67885){\circle*{0.001}}
\put(0.708669,0.68017){\circle*{0.001}}
\put(0.714829,0.68151){\circle*{0.001}}
\put(0.720993,0.68289){\circle*{0.001}}
\put(0.727161,0.6843){\circle*{0.001}}
\put(0.733333,0.68574){\circle*{0.001}}
\put(0.739507,0.68721){\circle*{0.001}}
\put(0.745684,0.68871){\circle*{0.001}}
\put(0.751862,0.69024){\circle*{0.001}}
\put(0.758042,0.6918){\circle*{0.001}}
\put(0.764224,0.69339){\circle*{0.001}}
\put(0.770405,0.69501){\circle*{0.001}}
\put(0.776588,0.69666){\circle*{0.001}}
\put(0.78277,0.69834){\circle*{0.001}}
\put(0.788951,0.70005){\circle*{0.001}}
\put(0.795131,0.70179){\circle*{0.001}}
\put(0.80131,0.70355){\circle*{0.001}}
\put(0.807487,0.70535){\circle*{0.001}}
\put(0.813662,0.70717){\circle*{0.001}}
\put(0.819834,0.70901){\circle*{0.001}}
\put(0.826003,0.71089){\circle*{0.001}}
\put(0.832169,0.71279){\circle*{0.001}}
\put(0.838331,0.71472){\circle*{0.001}}
\put(0.844489,0.71667){\circle*{0.001}}
\put(0.850642,0.71865){\circle*{0.001}}
\put(0.85679,0.72065){\circle*{0.001}}
\put(0.862932,0.72268){\circle*{0.001}}
\put(0.869069,0.72474){\circle*{0.001}}
\put(0.8752,0.72682){\circle*{0.001}}
\put(0.881324,0.72892){\circle*{0.001}}
\put(0.887442,0.73105){\circle*{0.001}}
\put(0.893552,0.7332){\circle*{0.001}}
\put(0.899655,0.73538){\circle*{0.001}}
\put(0.90575,0.73757){\circle*{0.001}}
\put(0.911837,0.73979){\circle*{0.001}}
\put(0.917915,0.74204){\circle*{0.001}}
\put(0.923984,0.7443){\circle*{0.001}}
\put(0.930044,0.74659){\circle*{0.001}}
\put(0.936095,0.7489){\circle*{0.001}}
\put(0.942135,0.75123){\circle*{0.001}}
\put(0.948166,0.75358){\circle*{0.001}}
\put(0.954185,0.75595){\circle*{0.001}}
\put(0.960195,0.75834){\circle*{0.001}}
\put(0.966192,0.76075){\circle*{0.001}}
\put(0.972179,0.76318){\circle*{0.001}}
\put(0.978154,0.76563){\circle*{0.001}}
\put(0.984117,0.7681){\circle*{0.001}}
\put(0.990067,0.77059){\circle*{0.001}}
\put(0.996005,0.7731){\circle*{0.001}}
\put(1.001931,0.77562){\circle*{0.001}}
\put(1.007843,0.77817){\circle*{0.001}}
\put(1.013741,0.78073){\circle*{0.001}}
\put(1.019626,0.78331){\circle*{0.001}}
\put(1.025498,0.7859){\circle*{0.001}}
\put(1.031354,0.78851){\circle*{0.001}}
\put(1.037197,0.79114){\circle*{0.001}}
\put(1.043025,0.79378){\circle*{0.001}}
\put(1.048838,0.79644){\circle*{0.001}}
\put(1.054636,0.79912){\circle*{0.001}}
\put(1.060418,0.80181){\circle*{0.001}}
\put(1.066185,0.80452){\circle*{0.001}}
\put(1.071937,0.80724){\circle*{0.001}}
\put(1.077672,0.80997){\circle*{0.001}}
\put(1.08339,0.81272){\circle*{0.001}}
\put(1.089093,0.81548){\circle*{0.001}}
\put(1.094779,0.81826){\circle*{0.001}}
\put(1.100447,0.82105){\circle*{0.001}}
\put(1.106099,0.82385){\circle*{0.001}}
\put(1.111734,0.82667){\circle*{0.001}}
\put(1.117351,0.82949){\circle*{0.001}}
\put(1.12295,0.83233){\circle*{0.001}}
\put(1.128532,0.83518){\circle*{0.001}}
\put(1.134096,0.83805){\circle*{0.001}}
\put(1.139641,0.84092){\circle*{0.001}}
\put(1.145168,0.84381){\circle*{0.001}}
\put(1.150677,0.8467){\circle*{0.001}}
\put(1.156167,0.84961){\circle*{0.001}}
\put(1.161638,0.85253){\circle*{0.001}}
\put(1.16709,0.85545){\circle*{0.001}}
\put(1.172524,0.85839){\circle*{0.001}}
\put(1.177937,0.86134){\circle*{0.001}}
\put(1.183332,0.86429){\circle*{0.001}}
\put(1.188707,0.86726){\circle*{0.001}}
\put(1.194062,0.87023){\circle*{0.001}}
\put(1.199398,0.87321){\circle*{0.001}}
\put(1.204713,0.8762){\circle*{0.001}}
\put(1.210009,0.8792){\circle*{0.001}}
\put(1.215284,0.88221){\circle*{0.001}}
\put(1.220539,0.88522){\circle*{0.001}}
\put(1.225774,0.88824){\circle*{0.001}}
\put(1.230988,0.89127){\circle*{0.001}}
\put(1.236182,0.8943){\circle*{0.001}}
\put(1.241355,0.89734){\circle*{0.001}}
\put(1.246507,0.90039){\circle*{0.001}}
\put(1.251638,0.90344){\circle*{0.001}}
\put(1.256749,0.9065){\circle*{0.001}}
\put(1.261838,0.90957){\circle*{0.001}}
\put(1.266906,0.91264){\circle*{0.001}}
\put(1.271953,0.91571){\circle*{0.001}}
\put(1.276979,0.9188){\circle*{0.001}}
\put(1.281983,0.92188){\circle*{0.001}}
\put(1.286966,0.92497){\circle*{0.001}}
\put(1.291927,0.92807){\circle*{0.001}}
\put(1.296867,0.93117){\circle*{0.001}}
\put(1.301785,0.93427){\circle*{0.001}}
\put(1.306682,0.93737){\circle*{0.001}}
\put(1.311557,0.94048){\circle*{0.001}}
\put(1.31641,0.9436){\circle*{0.001}}
\put(1.321241,0.94672){\circle*{0.001}}
\put(1.32605,0.94983){\circle*{0.001}}
\put(1.330838,0.95296){\circle*{0.001}}
\put(1.335603,0.95608){\circle*{0.001}}
\put(1.340347,0.95921){\circle*{0.001}}
\put(1.345068,0.96234){\circle*{0.001}}
\put(1.349768,0.96547){\circle*{0.001}}
\put(1.354445,0.96861){\circle*{0.001}}
\put(1.3591,0.97174){\circle*{0.001}}
\put(1.363733,0.97488){\circle*{0.001}}
\put(1.368344,0.97802){\circle*{0.001}}
\put(1.372933,0.98115){\circle*{0.001}}
\put(1.3775,0.98429){\circle*{0.001}}
\put(1.382044,0.98744){\circle*{0.001}}
\put(1.386566,0.99058){\circle*{0.001}}
\put(1.391066,0.99372){\circle*{0.001}}
\put(1.395544,0.99686){\circle*{0.001}}
\put(1.399999,1){\circle*{0.001}}
\put(0.733333,1.942809){\circle*{0.001}}
\put(0.723317,1.942761){\circle*{0.001}}
\put(0.71327,1.942618){\circle*{0.001}}
\put(0.703194,1.942378){\circle*{0.001}}
\put(0.693094,1.94204){\circle*{0.001}}
\put(0.68297,1.941601){\circle*{0.001}}
\put(0.672827,1.941062){\circle*{0.001}}
\put(0.662666,1.940421){\circle*{0.001}}
\put(0.652491,1.939676){\circle*{0.001}}
\put(0.642305,1.938827){\circle*{0.001}}
\put(0.632111,1.937872){\circle*{0.001}}
\put(0.621912,1.936811){\circle*{0.001}}
\put(0.611711,1.935642){\circle*{0.001}}
\put(0.601512,1.934364){\circle*{0.001}}
\put(0.591317,1.932977){\circle*{0.001}}
\put(0.58113,1.93148){\circle*{0.001}}
\put(0.570954,1.929871){\circle*{0.001}}
\put(0.560794,1.928151){\circle*{0.001}}
\put(0.550651,1.926318){\circle*{0.001}}
\put(0.54053,1.924372){\circle*{0.001}}
\put(0.530434,1.922313){\circle*{0.001}}
\put(0.520367,1.920139){\circle*{0.001}}
\put(0.510333,1.917851){\circle*{0.001}}
\put(0.500334,1.915449){\circle*{0.001}}
\put(0.490375,1.912931){\circle*{0.001}}
\put(0.480459,1.910298){\circle*{0.001}}
\put(0.47059,1.90755){\circle*{0.001}}
\put(0.460772,1.904686){\circle*{0.001}}
\put(0.451008,1.901707){\circle*{0.001}}
\put(0.441302,1.898614){\circle*{0.001}}
\put(0.431657,1.895405){\circle*{0.001}}
\put(0.422077,1.892081){\circle*{0.001}}
\put(0.412567,1.888644){\circle*{0.001}}
\put(0.403129,1.885092){\circle*{0.001}}
\put(0.39377,1.881427){\circle*{0.001}}
\put(0.38449,1.87765){\circle*{0.001}}
\put(0.37529,1.87376){\circle*{0.001}}
\put(0.36618,1.869759){\circle*{0.001}}
\put(0.35716,1.865648){\circle*{0.001}}
\put(0.34823,1.861428){\circle*{0.001}}
\put(0.3394,1.857099){\circle*{0.001}}
\put(0.33067,1.852662){\circle*{0.001}}
\put(0.32204,1.84812){\circle*{0.001}}
\put(0.31352,1.843472){\circle*{0.001}}
\put(0.30511,1.838721){\circle*{0.001}}
\put(0.29682,1.833868){\circle*{0.001}}
\put(0.28864,1.828914){\circle*{0.001}}
\put(0.28058,1.82386){\circle*{0.001}}
\put(0.27264,1.818709){\circle*{0.001}}
\put(0.26483,1.813463){\circle*{0.001}}
\put(0.25715,1.808122){\circle*{0.001}}
\put(0.24959,1.802688){\circle*{0.001}}
\put(0.24217,1.797164){\circle*{0.001}}
\put(0.23489,1.791552){\circle*{0.001}}
\put(0.22774,1.785853){\circle*{0.001}}
\put(0.22074,1.78007){\circle*{0.001}}
\put(0.21387,1.774204){\circle*{0.001}}
\put(0.20715,1.768258){\circle*{0.001}}
\put(0.20058,1.762234){\circle*{0.001}}
\put(0.19416,1.756134){\circle*{0.001}}
\put(0.18788,1.749961){\circle*{0.001}}
\put(0.18176,1.743717){\circle*{0.001}}
\put(0.17579,1.737404){\circle*{0.001}}
\put(0.16998,1.731025){\circle*{0.001}}
\put(0.16432,1.724581){\circle*{0.001}}
\put(0.15882,1.718077){\circle*{0.001}}
\put(0.15348,1.711513){\circle*{0.001}}
\put(0.1483,1.704894){\circle*{0.001}}
\put(0.14327,1.69822){\circle*{0.001}}
\put(0.13841,1.691495){\circle*{0.001}}
\put(0.13371,1.684721){\circle*{0.001}}
\put(0.12917,1.677901){\circle*{0.001}}
\put(0.1248,1.671038){\circle*{0.001}}
\put(0.12058,1.664134){\circle*{0.001}}
\put(0.11654,1.657191){\circle*{0.001}}
\put(0.11265,1.650213){\circle*{0.001}}
\put(0.10893,1.643201){\circle*{0.001}}
\put(0.10536,1.636159){\circle*{0.001}}
\put(0.10197,1.629089){\circle*{0.001}}
\put(0.09873,1.621993){\circle*{0.001}}
\put(0.09566,1.614875){\circle*{0.001}}
\put(0.09274,1.607736){\circle*{0.001}}
\put(0.08999,1.60058){\circle*{0.001}}
\put(0.0874,1.593408){\circle*{0.001}}
\put(0.08497,1.586223){\circle*{0.001}}
\put(0.08269,1.579028){\circle*{0.001}}
\put(0.08058,1.571825){\circle*{0.001}}
\put(0.07861,1.564616){\circle*{0.001}}
\put(0.07681,1.557403){\circle*{0.001}}
\put(0.07515,1.55019){\circle*{0.001}}
\put(0.07365,1.542979){\circle*{0.001}}
\put(0.0723,1.53577){\circle*{0.001}}
\put(0.0711,1.528568){\circle*{0.001}}
\put(0.07005,1.521373){\circle*{0.001}}
\put(0.06914,1.514188){\circle*{0.001}}
\put(0.06838,1.507016){\circle*{0.001}}
\put(0.06776,1.499857){\circle*{0.001}}
\put(0.06728,1.492715){\circle*{0.001}}
\put(0.06694,1.485591){\circle*{0.001}}
\put(0.06674,1.478487){\circle*{0.001}}
\put(0.06667,1.471404){\circle*{0.001}}

\end{picture}
\end{center}
\caption{This picture shows the regions for states and entanglement witnesses
on the two dimensional section through $\varrho_{E_1}$, $\varrho_{E_2}$ and $\varrho_*$ in the two qubit system.
There exist entanglement witnesses outside of the face $F_{E_2^\perp}$ since
$E_2$ has no product vector, but there exists no entanglement witnesses outside of $F_{E_1^\perp}$.
This is the real picture with respect to the
Hilbert-Schmidt norm.}
\end{figure}

First of all, we note that a two qubit Hermitian matrix is an entanglement witness if and only if
it is a Schmidt number $2$ witness if and only if it belongs to $\blockpos_1\setminus\blockpos_2$.
There exist entanglement witnesses outside of the face $F_{E_2^\perp}$ since there is  no product vector in $E_2$,
that is,
$$
S^\prime(E_2,1)=\emptyset.
$$
On the other hand, there exists no entanglement witness outside of
the faces $F_{E_1}$, $F_{E_1^\perp}$ or $F_{E_2}$, because
$S^\prime(E_1^\perp,1)$, $S^\prime(E_1,1)$ and $S^\prime(E_2^\perp,1)$ are
nonempty. In fact, we have
$$
\dim S(E_1^\perp,1)=2,\qquad
\# S^\prime (E_1,1)=1,\qquad
\dim S(E_2^\perp,1)=2,
$$
where $\# X$ denotes the cardinality of \emph{rays} in $X$.

We note that one dimensional subspaces in $\mathbb C^2\ot\mathbb C^2$ form the three dimensional complex projective space,
and subspaces spanned by Schmidt rank one vectors form a quadric surface, which give rise to non-generic cases.
Note that $E_2$ is a generic case, but $E_1$ is non-generic.
By {\sc Figure 4}, we note that $\blockpos_1\supsetneqq\blockpos_2$ is locally distinguishable around
$\varrho_{E_1}$, $\varrho_{E_1^\perp}$ and $\varrho_{E_2}$, but
it cannot be distinguished by principal supporting hyperplanes around $\varrho_{E_1^\perp}$
because
$$
S^\prime(E_1,2)=\emptyset,
$$
by Proposition \ref{pro_detect}.
Indeed, the boundary of $\mathcal D=\blockpos_2$ is tangent to the boundary of $\blockpos_1$ at $\varrho_{E_1^\perp}$
in {\sc Figure 4}. The inclusion
is locally distinguishable around $\varrho_{E_1}$ and $\varrho_{E_2}$ by principal supporting hyperplanes, since
$S^\prime(E_i^\perp,2)$ is nonempty for $i=1,2$.
In fact, we have
$$
\dim S(E_i^\perp,2)=3,\qquad i=1,2.
$$
We can see in {\sc Figure 4} that $\blockpos_1\supsetneqq\blockpos_2$ is locally distinguishable around $\varrho_{E_2}$
by principal supporting hyperplanes, but we cannot see it around $\varrho_{E_1}$ in this two dimensional section.
\ $\Box$
\medskip

We turn our attention to the bi-qutrit case of $m=n=3$. In this case, we have
\begin{itemize}
\item
$\kappa_{\dim E}=3$ for $\dim E=1$,
\item
$\kappa_{\dim E}=2$ for $\dim E=2,3,4$,
\item
$\kappa_{\dim E}=1$ for $\dim E=5,6,7,8$.
\end{itemize}
For a generic $2$ dimensional subspace $E$, we have
$$
S(E,1)=\{0\},\qquad \# S^\prime (E,2)= 3,\qquad \dim S(E,3)=2.
$$
For a generic $5$ dimensional subspace $E$, we also have
$$
\# S^\prime (E,1)=6,\qquad \dim S(E,2)= 4,\qquad \dim S(E,3)=5.
$$
We provide two dimensional subspaces
in the $3\ot 3$ system which behave as generic and non-generic ones, respectively. It may happen that $S^\prime(E,2)$ is empty even though
$S^\prime(E,\ell)$ is nonempty for $\ell=1,3$.

\medskip
\noindent
{\bf Example 2}.
We revisit examples in \cite{han_kye_2025_a}, which were main motivations for this paper, as they are shown in {\sc Figure 5}.
We take
$$
|\xi_1\ran=|01\ran,\qquad
|\xi_2\ran = {1\over \sqrt{2}}(|01\ran+|10\ran),
\qquad
|\xi_3\ran = {1\over \sqrt{3}}(|00\ran+|11\ran+|22\ran)
$$
in $\mathbb C^3\ot\mathbb C^3$, put $\varrho_j=|\xi_j\ran\lan\xi_j|$ for $j=1,2,3$,
and consider the two dimensional subspace $E_i$ spanned by $|\xi_3\ran$ and $|\xi_i\ran$
for $i=1,2$. See {\sc Figure 5}.

\begin{figure}
\begin{center}
\setlength{\unitlength}{2.5 truecm}
\begin{picture}(2.5,1.5)
\thicklines
\put(0,1.41421){\circle*{0.03}}%A
\put(0,0){\circle*{0.03}}%B
\put(0.80178,0.70710){\circle*{0.03}}%C
\put(0.62360,0.70710){\circle*{0.03}}%O
\put(0.70156,0.61871){\circle*{0.03}}%E
\put(0.70156,0.79549){\circle*{0.03}}%F
\put(0.93541,0.35355){\circle*{0.03}}%I
\put(0.74833,0.56568){\circle*{0.03}}%J

\dottedline{0.005}(0,1.41421356237309)(0,0)
\dottedline{0.005}(0,0)(0.801783725737273,0.707106781186547)
\dottedline{0.005}(0.801783725737273,0.707106781186547)(0,1.41421356237309)
\dottedline{0.05}(0,1.41421356237309)(0.935414346693485,0.353553390593274)
\dottedline{0.005}(0.935414346693485,0.353553390593274)(0,0)
\dottedline{0.005}(0,0)(0.748331477354788,0.565685424949238)

\dottedline{0.05}(0,0.71)(0.70156,0.71)
\dottedline{0.05}(1.5,0.71)(2.3,0.71)
\dottedline{0.05}(0,0)(0.70156,0.79549)
\dottedline{0.05}(1.5,0)(2.30178,0.90913)

\put(-0.15,1,4){$\varrho_3$}
\put(-0.15,0){$\varrho_1$}
\put(0.45,0.67){$\varrho_*$}
\put(0.81,0.77){$\varrho_{E_1^\perp}$}
\put(-0.22,0.77){$\varrho_{E_1}$}
\put(0,0.70710){\circle*{0.03}}%\varrho_{E_1}

\put(1.5,1.41421){\circle*{0.03}}%A
\put(1.5,0){\circle*{0.03}}%B
\put(2.30178,0.70710){\circle*{0.03}}%C
\put(2.12360,0.70710){\circle*{0.03}}%O
\put(2.20156,0.61871){\circle*{0.03}}%E
\put(2.20156,0.79549){\circle*{0.03}}%F
\put(2.30178,0.90913){\circle*{0.03}}%G
\put(2.43541,0.35355){\circle*{0.03}}%I
\put(2.24833,0.56568){\circle*{0.03}}%J

\dottedline{0.005}(1.5,1.41421356237309)(1.5,0)
\dottedline{0.005}(1.5,0)(2.30178372573727,0.707106781186547)
\dottedline{0.005}(2.30178372573727,0.707106781186547)(1.5,1.41421356237309)
\dottedline{0.05}(1.5,1.41421356237309)(2.43541434669349,0.353553390593274)

\dottedline{0.005}(1.5,1.41421356237309)(2.30178372573727,0.90913729009699)
\dottedline{0.005}(2.30178372573727,0.90913729009699)(2.43541434669349,0.353553390593274)
\dottedline{0.005}(2.43541434669349,0.353553390593274)(1.5,0)
\dottedline{0.005}(1.5,0)(2.24833147735479,0.565685424949238)

\put(1.34,1,4){$\varrho_3$}
\put(1.34,0){$\varrho_2$}
\put(1.96,0.67){$\varrho_*$}
\put(2.31,0.77){$\varrho_{E_2^\perp}$}
\put(1.26,0.77){$\varrho_{E_2}$}
\put(1.5,0.70710){\circle*{0.03}}%\varrho_{E_2}

%3x3_sn=1%BP1
\put(0.80178,0.70710){\circle*{0.001}}
\put(0.80648,0.70295){\circle*{0.001}}
\put(0.81144,0.69858){\circle*{0.001}}
\put(0.81667,0.69397){\circle*{0.001}}
\put(0.82221,0.68909){\circle*{0.001}}
\put(0.82807,0.68392){\circle*{0.001}}
\put(0.83428,0.67844){\circle*{0.001}}
\put(0.84089,0.67261){\circle*{0.001}}
\put(0.84792,0.66640){\circle*{0.001}}
\put(0.85543,0.65979){\circle*{0.001}}
\put(0.86345,0.65271){\circle*{0.001}}

\put(0.86345,0.65271){\circle*{0.001}}
\put(0.87147,0.64527){\circle*{0.001}}
\put(0.87893,0.63762){\circle*{0.001}}
\put(0.88586,0.62979){\circle*{0.001}}
\put(0.89228,0.62180){\circle*{0.001}}
\put(0.89821,0.61369){\circle*{0.001}}
\put(0.90367,0.60548){\circle*{0.001}}
\put(0.90869,0.59719){\circle*{0.001}}
\put(0.91329,0.58885){\circle*{0.001}}
\put(0.91749,0.58047){\circle*{0.001}}
\put(0.92131,0.57208){\circle*{0.001}}
\put(0.92477,0.56368){\circle*{0.001}}
\put(0.92789,0.55529){\circle*{0.001}}
\put(0.93069,0.54693){\circle*{0.001}}
\put(0.93319,0.53860){\circle*{0.001}}
\put(0.93541,0.53033){\circle*{0.001}}
\put(0.93736,0.52211){\circle*{0.001}}
\put(0.93905,0.51395){\circle*{0.001}}
\put(0.94051,0.50587){\circle*{0.001}}
\put(0.94175,0.49787){\circle*{0.001}}
\put(0.94277,0.48996){\circle*{0.001}}
\put(0.94361,0.48214){\circle*{0.001}}
\put(0.94426,0.47441){\circle*{0.001}}
\put(0.94473,0.46678){\circle*{0.001}}
\put(0.94505,0.45925){\circle*{0.001}}
\put(0.94522,0.45183){\circle*{0.001}}
\put(0.94525,0.44451){\circle*{0.001}}
\put(0.94515,0.43730){\circle*{0.001}}
\put(0.94493,0.43020){\circle*{0.001}}
\put(0.94459,0.42321){\circle*{0.001}}
\put(0.94415,0.41633){\circle*{0.001}}
\put(0.94361,0.40956){\circle*{0.001}}
\put(0.94299,0.40290){\circle*{0.001}}
\put(0.94227,0.39635){\circle*{0.001}}
\put(0.94148,0.38991){\circle*{0.001}}
\put(0.94062,0.38359){\circle*{0.001}}
\put(0.93969,0.37737){\circle*{0.001}}
\put(0.93870,0.37125){\circle*{0.001}}
\put(0.93765,0.36525){\circle*{0.001}}
\put(0.93656,0.35935){\circle*{0.001}}
\put(0.93541,0.35355){\circle*{0.001}}

%3x3_sn=1%BP2
\put(0.80178,0.70710){\circle*{0.001}}
\put(0.80852,0.69871){\circle*{0.001}}
\put(0.81167,0.69004){\circle*{0.001}}
\put(0.81249,0.68140){\circle*{0.001}}
\put(0.81177,0.67296){\circle*{0.001}}
\put(0.81000,0.66483){\circle*{0.001}}
\put(0.80754,0.65705){\circle*{0.001}}
\put(0.80461,0.64963){\circle*{0.001}}
\put(0.80139,0.64259){\circle*{0.001}}
\put(0.79799,0.63592){\circle*{0.001}}
\put(0.79449,0.62960){\circle*{0.001}}
\put(0.79095,0.62361){\circle*{0.001}}
\put(0.78742,0.61794){\circle*{0.001}}
\put(0.78392,0.61258){\circle*{0.001}}
\put(0.78048,0.60749){\circle*{0.001}}
\put(0.77711,0.60267){\circle*{0.001}}
\put(0.77382,0.59809){\circle*{0.001}}
\put(0.77062,0.59375){\circle*{0.001}}
\put(0.76751,0.58962){\circle*{0.001}}
\put(0.76449,0.58569){\circle*{0.001}}
\put(0.76157,0.58195){\circle*{0.001}}
\put(0.75874,0.57839){\circle*{0.001}}
\put(0.75601,0.57499){\circle*{0.001}}
\put(0.75336,0.57174){\circle*{0.001}}
\put(0.75080,0.56864){\circle*{0.001}}
\put(0.74833,0.56568){\circle*{0.001}}

%3x3_sn=2%BP2
\put(2.30178,0.70710){\circle*{0.001}}
\put(2.29990,0.69911){\circle*{0.001}}
\put(2.29792,0.69129){\circle*{0.001}}
\put(2.29587,0.68366){\circle*{0.001}}
\put(2.29375,0.67623){\circle*{0.001}}
\put(2.29158,0.66901){\circle*{0.001}}
\put(2.28937,0.66199){\circle*{0.001}}
\put(2.28712,0.65519){\circle*{0.001}}
\put(2.28484,0.64860){\circle*{0.001}}
\put(2.28255,0.64222){\circle*{0.001}}
\put(2.28026,0.63605){\circle*{0.001}}
\put(2.27796,0.63009){\circle*{0.001}}
\put(2.27568,0.62433){\circle*{0.001}}
\put(2.27340,0.61878){\circle*{0.001}}
\put(2.27114,0.61342){\circle*{0.001}}
\put(2.26891,0.60825){\circle*{0.001}}
\put(2.26670,0.60326){\circle*{0.001}}
\put(2.26451,0.59845){\circle*{0.001}}
\put(2.26236,0.59382){\circle*{0.001}}
\put(2.26025,0.58935){\circle*{0.001}}
\put(2.25817,0.58504){\circle*{0.001}}
\put(2.25612,0.58088){\circle*{0.001}}
\put(2.25412,0.57687){\circle*{0.001}}
\put(2.25215,0.57301){\circle*{0.001}}
\put(2.25022,0.56928){\circle*{0.001}}
\put(2.24833,0.56568){\circle*{0.001}}

\end{picture}
\end{center}
\caption{
These pictures for the bi-qutrit system show the regions for states, Schmidt number $3$ witnesses and Schmidt number $2$ witnesses
on the $2$-dimensional sections through $\varrho_3$, $\varrho_i$ and $\varrho_*$, where $\varrho_j=|\xi_j\ran\lan\xi_j|$ and
$E_i$ is the two dimensional space spanned by $|\xi_3\ran$ and $|\xi_i\ran$ of Schmidt ranks $3$ and $i$, for $i=1,2$.
This is the real picture with respect to the Hilbert-Schmidt norm.
}
\end{figure}
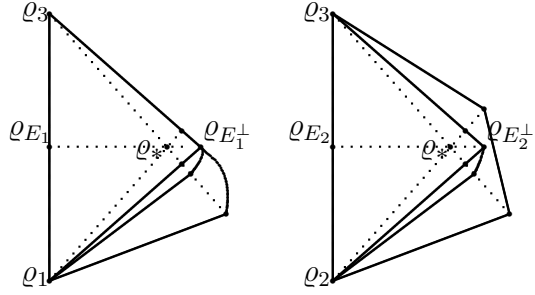

Two dimensional subspaces in $\mathbb C^3\ot\mathbb C^3$ form the Grassmann manifold $\Gr(2,9)$ of
complex dimension $14$.
Generic two dimensional subspaces $E$ satisfy $\# S^\prime(E,2)=3$ and $S^\prime(E,1)=\emptyset$.
The $3\times 3$ matrix corresponding to $a|\xi_1\ran+b|\xi_3\ran$ has determinant $(b/\sqrt 3)^3$, and it has rank one when $b=0$.
Therefore, the space $E_1$ is non-generic, and we have
$$
\# S^\prime(E_1,1)=1,\qquad
\# S^\prime(E_1,2)=0,\qquad
\# S^\prime(E_1,3)=\infty.
$$
Therefore, there exists no Schmidt number witness outside of the face $F_{E_1^\perp}$ in the left picture of
{\sc Figure 5}. We also see that the inclusion $\blockpos_1\supsetneqq\blockpos_2$ cannot be distinguished
around $\varrho_{E_1^\perp}$ by principal supporting hyperplanes, but
$\blockpos_2\supsetneqq\blockpos_3$ is locally distinguishable around $\varrho_{E_1^\perp}$ by principal
supporting hyperplanes. Nevertheless, $\blockpos_1\supsetneqq\blockpos_2$ is also locally distinguishable around
$\varrho_{E_1^\perp}$ by Theorem \ref{thm_general} below.

On the other hand, the $3\times 3$ matrix corresponding to $a|\xi_2\ran+b|\xi_3\ran$ has determinant
$\frac b{\sqrt 3}\left(\frac {b^2}3-\frac{a^2}2\right)$, which has three simple zeros, and the corresponding matrices have rank two.
Therefore, $E_2$ behaves like a generic subspace, and we have
$$
\# S^\prime(E_2,1)=0,\qquad
\# S^\prime(E_2,2)=3,\qquad
\# S^\prime(E_2,3)=\infty.
$$
First of all, $S^\prime(E_2,1)=\emptyset$ implies that there exist Schmidt number $2$ witnesses outside
of the face $F_{E_2^\perp}$.
There exist no Schmidt number $3$ witnesses outside of the face $F_{E_2^\perp}$ since
$S^\prime(E_2,2)$ is nonempty. But, the inclusion
$\blockpos_2\supsetneqq\blockpos_3$ is locally distinguishable around $\varrho_{E_2^\perp}$
by principal supporting hyperplanes since  $S^\prime(E_2,3)\neq\emptyset$, and so there exist Schmidt number $3$ witnesses
around $\varrho_{E_2^\perp}$.
\ $\Box$

%%%%%%%%%%%%%%%%%%%%%%%%%%%%%%%%%%%%%%%%%%%%%%%%%%
\section{Schmidt number witnesses around $\varrho_{E^\perp}$ in general}\label{sec_general}

In this section, we show that every neighborhood of $\varrho_{E^\perp}$ contains Schmidt number $\ell$ witnesses
for every $\ell\ge\mu(E)$, without any genericity assumption.
Throughout this section, $E$ is a nonzero proper subspace of $\mathbb C^m\ot\mathbb C^n$, $D=\dim E^\perp$,
and $C_k$ is the closed cone of all vectors with Schmidt rank $\le k$ as in the proof of Theorem \ref{main}.
We use the Hilbert space norm $\|\cdot\|$ for vectors and the Hilbert-Schmidt norm for matrices.
We begin with the following elementary observation.

\begin{lemma}\label{lem_closure}
Suppose that $k<m\meet n$. For every $|\xi\ran\in\mathbb C^m\ot\mathbb C^n$ with $\SR|\xi\ran\le k$ and every $\delta>0$,
there exists $|v\ran$ with $\SR|v\ran=k+1$ and $\||v\ran-|\xi\ran\|<\delta$.
\end{lemma}

\begin{proof}
Let $|\xi\ran=\sum_{i=1}^r s_i|e_i\ran|f_i\ran$ be the Schmidt decomposition with $r=\SR|\xi\ran\le k$, and extend
$\{|e_i\ran\}$ and $\{|f_i\ran\}$ to orthonormal bases of $\mathbb C^m$ and $\mathbb C^n$, respectively.
Since $k+1\le m\meet n$, the vector $|v\ran=|\xi\ran+t\sum_{i=r+1}^{k+1}|e_i\ran|f_i\ran$ has Schmidt rank $k+1$ for $t>0$,
and $\||v\ran-|\xi\ran\|=t\sqrt{k+1-r}$.
\end{proof}

The following theorem shows that whenever $\varrho_{E^\perp}$ lies on the boundary of $\blockpos_k$,
every neighborhood of $\varrho_{E^\perp}$ in $\mathcal H$ contains Schmidt number $k+1$ witnesses.

\begin{theorem}\label{thm_general}
Let $E$ be a nonzero proper subspace of $\mathbb C^m\ot\mathbb C^n$ and $k<m\meet n$.
If $S(E,k)\neq\{0\}$, that is, $\varrho_{E^\perp}$ lies on the boundary of $\blockpos_k$, then the inclusion
$\blockpos_k\supsetneqq\blockpos_{k+1}$ is locally distinguishable around $\varrho_{E^\perp}$.
\end{theorem}

\begin{proof}
We note that $\varrho_{E^\perp}$ is a common boundary point of $\blockpos_k$ and $\blockpos_{k+1}$ by (\ref{boundary}).
Let $0<\delta<1/2$. Take a unit vector $|\xi_0\ran\in S(E,k)$. By Lemma \ref{lem_closure}, there exists $|v_0\ran$
with $\SR|v_0\ran=k+1$ and $\||v_0\ran-|\xi_0\ran\|<\delta$. We write $|v_0\ran=|x_0\ran+|\eta_0\ran$ with $|x_0\ran\in E$ and $|\eta_0\ran\in E^\perp$.
Then $\||\eta_0\ran\|=\|P_{E^\perp}(|v_0\ran-|\xi_0\ran)\|<\delta$ and $\||x_0\ran\|>1-\delta>1/2$.
Put $|x\ran=|x_0\ran/\||x_0\ran\|$, $|\eta\ran=|\eta_0\ran/\||x_0\ran\|$ and $|v\ran=|x\ran+|\eta\ran$.
Then $\||x\ran\|=1$, $\||\eta\ran\|<2\delta$ and $\SR|v\ran=k+1$. Since $C_k$ is closed and $|v\ran\notin C_k$, we have
$$
r=\dist(|v\ran,C_k)=\inf\{\||v\ran-|w\ran\|: |w\ran\in C_k\}>0.
$$
We choose $\beta,\gamma\in(0,\delta)$ with $\gamma(1+\beta)<\beta r^2$, and define the Hermitian matrix
\begin{equation}\label{Q}
Q=P_{E^\perp}-|\eta\ran\lan x|-|x\ran\lan\eta|+(\|\eta\|^2-\gamma)|x\ran\lan x|+\beta\left(P_E-|x\ran\lan x|\right).
\end{equation}
Since $|x\ran\in E$ is a unit vector and $|\eta\ran\in E^\perp$, a direct computation gives
\begin{equation}\label{Q_form}
\lan w|Q|w\ran=\left\|P_{E^\perp}|w\ran-\lan x|w\ran|\eta\ran\right\|^2+\beta\left\|P_E|w\ran-\lan x|w\ran|x\ran\right\|^2-\gamma|\lan x|w\ran|^2
\end{equation}
for every $|w\ran\in\mathbb C^m\ot\mathbb C^n$.

We first show that $Q$ is $k$-blockpositive. Suppose that $\SR|w\ran\le k$ and $\lan w|Q|w\ran<0$.
By (\ref{Q_form}), we have $c=\lan x|w\ran\neq 0$. Then $|w^\prime\ran=c^{-1}|w\ran$ satisfies $\SR|w^\prime\ran\le k$,
$\lan x|w^\prime\ran=1$ and $\lan w^\prime|Q|w^\prime\ran<0$, and so
$$
\left\|P_{E^\perp}|w^\prime\ran-|\eta\ran\right\|^2+\beta\left\|P_E|w^\prime\ran-|x\ran\right\|^2<\gamma.
$$
Since $P_E|v\ran=|x\ran$ and $P_{E^\perp}|v\ran=|\eta\ran$, we have
$$
\||w^\prime\ran-|v\ran\|^2=\left\|P_E|w^\prime\ran-|x\ran\right\|^2+\left\|P_{E^\perp}|w^\prime\ran-|\eta\ran\right\|^2
<\frac\gamma\beta+\gamma=\frac{\gamma(1+\beta)}\beta<r^2,
$$
which contradicts $r=\dist(|v\ran,C_k)$. Therefore, $Q\in\blockpos_k$.

Next, we have $\lan x|v\ran=1$, $P_{E^\perp}|v\ran=|\eta\ran$ and $P_E|v\ran=|x\ran$, and so (\ref{Q_form}) gives
$\lan v|Q|v\ran=-\gamma<0$. Since $\SR|v\ran=k+1$, we have $Q\notin\blockpos_{k+1}$.

Finally, we have $\tr Q=D+\|\eta\|^2-\gamma+\beta(\dim E-1)\ge D-\gamma>0$, and so
$W=Q/\tr Q$ belongs to $\mathcal H$ and is a Schmidt number $k+1$ witness. By (\ref{Q}), we have
$$
\|Q-P_{E^\perp}\|\le 2\|\eta\|+\|\eta\|^2+\gamma+\beta\sqrt{\dim E-1}<\delta\left(5+4\delta+\sqrt{\dim E-1}\right),
$$
and $|\tr Q-D|<\delta(4\delta+\dim E)$. Therefore, $W$ tends to $P_{E^\perp}/D=\varrho_{E^\perp}$ as $\delta\to 0$.
\end{proof}

\begin{corollary}\label{cor_iff}
Let $E$ be a nonzero proper subspace of $\mathbb C^m\ot\mathbb C^n$, and let $2\le\ell\le m\meet n$.
Then every neighborhood of $\varrho_{E^\perp}$ in $\mathcal H$ contains Schmidt number $\ell$ witnesses
if and only if $\ell\ge\mu(E)$.
\end{corollary}

\begin{proof}
Put $\mu=\mu(E)$. If $\ell<\mu$, then $S(E,\ell)=\{0\}$ and $\varrho_{E^\perp}$ is an interior point of $\blockpos_\ell$
by (\ref{boundary}). Therefore, a neighborhood of $\varrho_{E^\perp}$ is contained in $\blockpos_\ell$,
and it contains no element of $\blockpos_{\ell-1}\setminus\blockpos_\ell$.

We also note that $\varrho_{E^\perp}$ is an interior point of $\blockpos_{\mu-1}$,
but a boundary point of $\blockpos_\mu$. Considering the line segment from $\varrho_*$ to $\varrho_{E^\perp}$,
we see that the claim holds when $\ell=\mu$. Theorem 5.2 also shows the claim when $\ell>\mu$.
\end{proof}

Suppose that $S(E,k)\neq\{0\}$. When $S^\prime(E,k+1)\neq\emptyset$, or equivalently, $S(E,k)\subsetneqq S(E,k+1)$,
the inclusion in Theorem \ref{thm_general} is already detected by principal supporting
hyperplanes by Proposition \ref{pro_detect}. If $S(E,k)= S(E,k+1)$, then
we have seen that $\blockpos_k$ and $\blockpos_{k+1}$ have exactly the same supporting hyperplanes
through $\varrho_{E^\perp}$,
and so the inclusion $\blockpos_k\supsetneqq\blockpos_{k+1}$
cannot be seen to the first order at $\varrho_{E^\perp}$.
The witnesses in Theorem \ref{thm_general} come from the non-linear geometry of the boundary of $\blockpos_{k+1}$ at $\varrho_{E^\perp}$.

%%%%%%%%%%%%%%%%%%%%%%%%%%%%%%%%%%%%%%%%%%%%%%%%%%
\section{Conclusion}

In this paper, we have considered local geometry around a projection state $\varrho_{E^\perp}$ onto the subspace
$E^\perp\subset\mathbb C^m\ot\mathbb C^n$, when $\varrho_{E^\perp}$ is on the boundary of the convex set $\blockpos_k$
of all $k$-blockpositive matrices. We have shown that every neighborhood of $\varrho_{E^\perp}$ contains Schmidt number $\ell$
witnesses if and only if $\ell\ge\mu(E)$. Together with the result in \cite{han_kye_2025_a}, Schmidt number $\ell$ witnesses
appear outside of the face $F_{E^\perp}$ exactly when $\ell\le\mu(E)$, and arbitrarily close to its center $\varrho_{E^\perp}$
exactly when $\ell\ge\mu(E)$.

In order to understand the local shape of $\blockpos_k$ around $\varrho_{E^\perp}$,
we introduced the notion of the principal supporting hyperplanes,
which are determined by vectors in $E$ with Schmidt rank at most $k$. The inclusion $\blockpos_k\supsetneqq\blockpos_{k+1}$
is detected by supporting hyperplanes through $\varrho_{E^\perp}$ if and only if $S^\prime(E,k+1)$ is nonempty.
Using techniques from algebraic geometry, we have shown that this is the case for generic subspaces $E$, but not in general.

When $S^\prime(E,k+1)=\emptyset$, the convex sets $\blockpos_k$ and $\blockpos_{k+1}$ have the same supporting hyperplanes through
$\varrho_{E^\perp}$, and the Schmidt number $k+1$ witnesses around $\varrho_{E^\perp}$ come from the non-linear geometry of their boundaries.
It would be interesting to describe this non-linear geometry, for example the second order behavior of the boundaries of
$\blockpos_k$ and $\blockpos_{k+1}$ at $\varrho_{E^\perp}$, in terms of the geometry of the variety $S(E,k)$.

\end{document}